\documentclass[runningheads]{llncs}
\usepackage{amsmath}
\usepackage{subcaption}
\usepackage{xcolor}
\usepackage{comment}
\usepackage{amssymb}
\usepackage[lambda, advantage, adversary, probability, sets, logic, ff,primitives, asymptotics, keys, oracles, operators]{cryptocode}

\createprocedureblock{procb}{center, boxed}{}{}{}
\createpseudocodeblock{pcb}{center, boxed}{}{}{}

\newcommand{\commitgen}{{\sf CKeyGen}}
\newcommand{\commitkey}{{\sf ctk}}
\newcommand{\commit}{{\sf Com}}
\newcommand{\commitment}{\mathsf{cm}}
\newcommand{\pact}{{\sf PACT}}
\newcommand{\pactsetup}{{\sf PACT.Setup}}
\newcommand{\pactkeygen}{{\sf PACT.KeyGen}}
\newcommand{\pactmint}{{\sf PACT.Mint}}
\newcommand{\pactspend}{{\sf PACT.Spend}}
\newcommand{\pactconsent}{{\sf PACT.Endorse}}
\newcommand{\pactverify}{{\sf PACT.Verify}}

\newcommand{\pactpk}{\mathsf{pk}}

\newcommand{\pactaux}{\mathsf{aux}}
\newcommand{\pactpks}{\mathcal{PK}} 
\newcommand{\pactpklist}{\mathcal{L}_{\sf{PK}}}
\newcommand{\pactsk}{\mathsf{sk}} 
\newcommand{\pactproofs}{\mathsf{\pi}} 
\newcommand{\pactpolicy}{\mathsf{Pol}} 
\newcommand{\pacttx}{\mathsf{tx}}
\newcommand{\pacttxs}{\mathcal{T}}
\newcommand{\pacttxlist}{\mathcal{L}_{\sf{T}}}
\newcommand{\pactcoinvalue}{\mathsf{v}} 
\newcommand{\pactvalue}{\mathcal{V}} 
\newcommand{\coinrand}{\mathsf{r}} 
\newcommand{\cell}{\mathsf{Cell}} 
\newcommand{\token}{\mathsf{tk}} 
\newcommand{\dataend}{\mathsf{En}}
\newcommand{\extract}{\mathsf{extract}} 
\newcommand{\broadcastfunction}{\sf{BCast.GetProofs}} 
\newcommand{\fillasset}{\sf{P.GetEndorsement}} 

\newcommand{\expaccountlist}{\mathcal{L}_{\sf{ACC}}}
\newcommand{\expcorruptlist}{\mathcal{L}_{\sf{C}}}

\newcommand{\exptxlist}{\mathcal{L}_{\sf{C}_{TX}}}
\newcommand{\expaddacc}{{\sf Exp.AddAcc}}
\newcommand{\expspend}{{\sf Exp.Spend}}
\newcommand{\expcorrupt}{{\sf Exp.Corrupt}}
\newcommand{\exppolicy}{{\sf Exp.Policy}}
\newcommand{\expledger}{{\sf Exp.Latest}}
\newcommand{\exppost}{{\sf Exp.PostTx}}

\newcommand{\tab}{\ \ }
\newcommand{\group}{{\mathbb{G}}}
\newcommand{\grouporder}{{\mathbb{Z}_P}}
\newcommand{\leftbrc}{{\textnormal{(}}}
\newcommand{\rightbrc}{{\textnormal{)}}}
\newcommand{\ground}{{\perp}}
\renewcommand{\oracle}{\mathcal{O}}
\newcommand{\ret}{\sf{return\ }}
\newcommand{\msg}{{\sf m}}
\newcommand{\publicpar}{\sf{pp}} 
\newcommand{\Adv}{\mathbf{Adv}} 
\newcommand{\stat}{\mathsf{stat}}

\newcommand{\opone}{\star} 
\newcommand{\optwo}{\circ} 

\newcommand{\zkpok}{\sf{ZKPoK}} 
\newcommand{\nizksetup}{\sf{ZKSetup}}
\newcommand{\nizkprove}{\sf{ZKProve}} 
\newcommand{\nizkverify}{\sf{ZKVerify}} 
\newcommand{\nizkproof}{\pi_{\sf{ZKP}}} 
\newcommand{\crs}{\sf{crs}} 
\newcommand{\stmt}{\sf{stmt}} 
\newcommand{\wit}{\sf{wit}} 
\newcommand{\challengespace}{\mathcal{C}} 
\newcommand{\lang}{\mathcal{L}}
\newcommand{\relation}{\mathcal{R}} 

\newcommand{\positive}{2^N} 
\newcommand{\consistency}{\mathsf{C}} 
\newcommand{\equivalence}{\mathsf{EQ}} 
\newcommand{\asset}{\mathsf{A}} 
\newcommand{\balance}{\mathsf{BA}} 

\newcommand{\assetset}{{\mathcal{AS}}}

\newcommand{\hidingadv}{{\sf Hid_{}^{\adv}}(\secpar)}
\newcommand{\hidingadvped}{{\sf Hid_{HCOM}^{\adv}}(\secpar)}
\newcommand{\bindingadv}{{\sf Bind_{}^{\adv}}(\secpar)}

\newcommand{\ExpAnon}{{\sf ANON}}
\newcommand{\ExpAnonFull}{{\sf ANON}^{\adv}}
\newcommand{\AdvAnon}{\Adv_{\sf \ExpAnon, \pact}}
\newcommand{\ExpBalance}{{\sf Integrity}}
\newcommand{\ExpBalanceFull}{{\sf Integrity}^{\adv}}
\newcommand{\AdvBalance}{\Adv_{\sf \ExpBalance, \pact}}
\usepackage{tabularx}
\usepackage{makecell}
\usepackage{float}
\usepackage[hidelinks]{hyperref}
\usepackage{multirow}

\usepackage[T1]{fontenc}
\usepackage{graphicx}

\title{Private, Auditable, and Distributed Ledger for Financial Institutes}
\author{Shaltiel Eloul$^*$ \and  Yash Satsangi$^*$\and  Yeoh Wei Zhu$^*$ \and Omar Amer \and Georgios Papadopoulos \and  Marco Pistoia}
\institute{Global-Technology Applied-Research, JP Morgan Chase, London UK} 

\begin{document}

\maketitle
\def\thefootnote{*}\footnotetext{These authors contributed equally to this work}

\begin{abstract}
Distributed ledger technology offers several advantages for banking and finance industry, including efficient transaction processing and cross-party transaction reconciliation. The key challenges for adoption of this technology in financial institutes are (a) the building of a privacy-preserving ledger, (b) supporting auditing and regulatory requirements, and (c) flexibility to adapt to complex use-cases with multiple digital assets and actors. This paper proposes a framework\def\thefootnote{\textdagger}\footnote{\href{https://github.com/jpmorganchase/PADL}{\underline{PADL-source-code}} 
}.   for a private, audit-able, and distributed ledger (PADL) that adapts easily to fundamental use-cases within financial institutes
PADL employs widely-used cryptography schemes combined with zero-knowledge proofs to propose a transaction scheme for a `table' like ledger. It enables fast confidential peer-to-peer multi-asset transactions, and transaction graph anonymity, in a no-trust setup, but with customized privacy. We prove that integrity and anonymity of PADL is secured against a strong threat model. Furthermore, we showcase three fundamental real-life use-cases, namely, an assets exchange ledger, a settlement ledger, and a bond market ledger. Based on these use-cases we show that PADL supports smooth-lined inter-assets auditing while preserving privacy of the participants. For example, we show how a bank can be audited for its liquidity or credit risk without violation of privacy of itself or any other party, or how can PADL ensures honest coupon rate payment in bond market without sharing investors values. Finally, our evaluation shows PADL's advantage in performance against previous relevant schemes. 
\def\thefootnote{\arabic{footnote}}
\end{abstract}

\section{Introduction}
Blockchain and distributed ledgers technology has opened-up a new form of financial interaction, to transact with no requirement of mediator bodies. The impact is clearly evident in the continuous growth of public blockchain domains such as Bitcoin \cite{nakamoto2008}, Ethereum \cite{wood2014ethereum} and other ledgers based on smart contract platforms, notably HyperLedger~\cite{androulaki2018hyperledger}. 
For financial industry, which is originally centralized, adapting distributed ledger technology potentially offers several advantages, such as reducing efforts involved in settlement and post-trades processing, streamlining laborious auditing processes, shifting to cloud environments, and enabling nearly instant approval of transactions. However, the current public blockchain ecosystem does not necessarily match the needs of financial institutions that must comply with government regulations and laws. A specific example is that of privacy: public blockchains requires that all participants have access to the transaction for its validation. This is unacceptable (or even illegal) to majority of institutions that require their clients' personal information or trading strategy be kept private and safe.

In recent years the focus on `private' blockchains has increased to address concerns on privacy in public blockchain, which are considered only pseudo-anonymous.
Privacy may be enhanced on blockchain by encrypting or anonymizing certain information, for example, the transaction's values, participants addresses, type of assets, or transaction graph. Inevitably though, encryption of information adds to the complexity of transaction validation and auditing. While encryption may help with enabling private transaction, it makes it nearly impossible for financial institutions to meet trading requirements or answer auditor's queries without revealing sensitive information.

Zero-knowledge proofs (ZKP) introduced as interactive~\cite{goldwasser2019knowledge} or non-interactive protocols (NIZK)~\cite{blum2019non} between a \textit{prover} and \textit{verifier} to convince the truth of a statement without revealing any further information. This technology can be used in private ledger to validate integrity of transactions and answer auditor queries without revealing sensitive information. 
ZKP systems have seen rapid advancements, particularly in terms of efficiency \cite{groth2010short}, succinctness, and general applicability. These developments build upon foundational frameworks such as zk-SNARKs~\cite{sasson2014zksnarks} and extend to transparent systems that do not require a trusted setup, such as zk-STARKs~\cite{sasson2018zkstarks}, Ligero~\cite{bhadauria2020ligero++,ames2017ligero}, and Bulletproofs~\cite{SP:BBBPWM18}. One of the primary advantages of modern proof systems is their ability to produce short proofs and facilitate the batching verification of range proofs, which demonstrate that a hidden value lies within a specified interval.
Evidently, these frameworks can be combined with public blockchains to enhance the privacy of these blockchains with ZKP~\cite{sasson2014zerocash,sun2021survey}. 
However current efforts focused on public blockchains fail to address institutions or enterprise blockchain applications that have a different set of auditing and privacy requirements or constraints that are not the focus of public crypto-assets. 
In practice, it is also recognized that adapting to new cryptography tools may require far more time in a regulated environment where security vulnerability cannot be compromised. Finally, financial industry use-cases involve multiple actors and multiple assets or asset-type in the trade or transaction. Existing public blockchain hardly address the issue of maintaining or exchanging multiple assets, not to mention, inter-asset auditing or customizing privacy. Hence, a private distributed ledger for financial institutions must be considerate of the following requirements: 

\textbf{Auditing.} Auditing inevitably requires disclosing information to assess and limit the financial and security risks institutes possess. Auditing in a private ledger can be divided into two types. First is privacy preserving auditing which does not require violating privacy of other participants (if they are not directly involved). In this optimal private auditing, a $prover$ shares a ZKP to convince an auditor with the information encrypted to ensure; traceability, transaction integrity, or regulating trading requirements. This privacy preserving auditing would open-up in future the capability to audit financial institutes without laborious efforts or book disclosure, for example, maintaining level of capital at risk in a confidential multi-assets ledgers as we show here.  
Additionally, some regulations require opening all values. Hence `full' auditing in essence, letting only a specified party, the capability to open an encrypted value. The requirement here is that this type of auditing should be limited to the specified party. We later discuss the case of a `settlement bank', where tracing values can be done by a settlement bank for `full' auditing.

\textbf{Private Multi-Asset Ledger.} Public blockchain ledgers, in their original implementation use anonymous accounts,  but values, transactions graphs, balances, and assets are public. By private ledger, we consider this information fully or partially to be hidden with the combination of encryption and zero-knowledge proofs to create private blockchains. Note, that anonymity includes also the hiding of the transaction graph and their assets. Transaction graph can be an extremely resourceful information \cite{alarab2020competence,zhi2022ledgit}. From trading perspective, even a small history of transaction can be sufficient to reveal its strategy or position, however, from auditing perspective, it is still required to maintain traceability, either in privacy preserving manner, with using ZKP, or with decryption of information by an auditor.
Furthermore, majority of existing financial use-cases involve multiple assets and multiple actors. For example, existing trading applications support multiple equities, exchanges, loans or other complex financial instruments. Hence, confidential multi-asset and audit-able privacy between assets is a necessity in order to handle practical use-cases within financial institutions. 

\textbf{No-trust Setup.} No trusted setup is highly desired in the context of external auditing using ZKP. Some distributed ledgers may offer auditing capabilities with the help of a trusted entity such as third-party auditors. This indeed does not provide data protection. Moreover, trusted setup or shared setup, using original zk-SNARK systems for example, would require the auditor to trust the setup. Hence, regulatory auditing would require transparent proofs, i.e. without verification or proving keys. Even if the auditor holds or trusts the holders of the keys, it makes auditing an issue when inter-banking applications may have several and different auditing bodies, where sharing the keys with the bodies, compromises the privacy of the data. 
The trusted setup, in the auditing context, is hence highly undesirable to most financial institutions as a pre-requirement.
% , since it creates a single point of vulnerability with severely bad consequences if exploited.  PADL encryption and ZKP system is transparent and does not required trusted setup \ref{cryptography_toolbox}.  We show however, that PADL can easily be customised to the requirements of customised privacy, for example in settlement bank use case (among others) which can act as a trusted entity, but still in transparent manner.
%In that case, parties still do not share the keys, but generate signatures to the settler with the information required.  
%Practically, in a regulated environment, the cryptographic methods for proofs must be well-tested and complied by various bodies involved in the process, hence the auditors can be banks, brokers, custodians, and/or regulator bodies. 

In this paper, we propose PADL which follows a `table' ledger scheme and combines encrypted commitments with audit-tokens and NIZK to support private peer-to-peer multi-asset transaction. It supports privacy preserving auditing of inter-asset balances to prove rates and liquidity risks of financial institutions with no trusted setup, and supports full auditing for provable traceability of transactions. We show that with PADL framework, it is straight-forward to build financial market ledger such as bond markets, confidential asset swap, settlement bank with customized privacy.
The transaction scheme enables asynchronous proof generations amongst the participants, and provides a general transaction porpuse for various types of multi-asset transactions.
Furthermore, PADL is an open-source package, compatible with smart-contract to showcase its integration to various blockchain integration. Finally, our performance evaluation shows that PADL, is not only more flexible, but outperforms other private `table' ledgers without sacrificing privacy or auditing capabilities in any form. 
To summarize, our contributions are as follows:
\begin{itemize}
    \item \textbf{PADL framework.} We propose a complete and industry-ready distributed ledger for private and anonymous multi-asset transactions while supporting privacy-preserving and full auditing.
    \item \textbf{New privacy-preserving audits notions.} We identified several core privacy preserving auditing proofs that enables operational risk control and are crucial for financial institution operations. 
    \item \textbf{Evaluations.} We provide a comprehensive theoretical and empirical evaluation of PADL. We evaluate the security and privacy properties, and show its practicality via performance evaluation.
    \item \textbf{Usecase Showcase.} We showcase practical use-cases for financial institution ledgers that are enabled with PADL, thereby demonstrating the capability of  PADL in constructing complex transaction scheme for financial applications.

\end{itemize}

\section{Related Work}
Our work is closest to zkLedger, that is also a table based private distributed ledger that combines NIZK based on $\Sigma$-protocols \cite{cramer1996modular} with Pedersen commitment to propose a private and audit-able transaction scheme. However, zkLedger is limited to single non-private asset and thus it cannot support complex multi-asset applications that are common within finance industry. PADL's flexible transaction scheme supports confidential multi-asset transactions with inter-asset auditing, making it possible to use for different use-cases such as asset exchange, trading loan ledger, settlement ledger, bond market, or secondary markets. Not only this, as shown in our evaluation, PADL transaction scheme scales better with number of assets and participants compared to zkLedger's transaction scheme.

Other private ledgers that offer privacy and auditing include Solidus \cite{cecchetti2017solidus}, PEReDi \cite{aggelos2022peredi} Miniledger 
\cite{chatzigiannis2021miniledger}, Fabzk \cite{kang2019fabzk}, Azeroth \cite{jeong2023azeroth}, etc. Some of these ledger do not necessarily permit auditing while maintaining privacy \cite{gao2019private}, 
or that their auditing is possible only if the ledger is open to public or to a trusted third party, or a fixed auditor \cite{Androulaki2019}. Furthermore, these ledgers are focussed and restricted to single asset transactions and auditing.
Similarly, Platypus shows an `e-cash' scheme for meeting regulations under privacy~\cite{Wust2021}, but with the trust of a central bank. 
This can be a significant limitation of privacy, since auditor in such cases might be able to access private information related to banks and entities that are not the target of the audit or have different regulatory requirements.
Miniledger \cite{chatzigiannis2021miniledger}, simply presents a different implementation of the zkLedger whilst losing some capabilities in order to simplify it, but does not offer additional privacy value. PEReDi \cite{aggelos2022peredi} does not support multi-asset swaps and/or offline transactions.
Many privacy preserving protocols based on membership proof, such as on Monero \cite{noether2016ring} and ZCash \cite{sasson2014zerocash} are intended for efficiency in much larger number of participants than in the case of enterprise blockchain. Therefore, they are concerned mostly with single asset/currency and (pseudo) anonymity without auditing or without multi-assets capability. However recent works investigate and offer Monero protocols with tractability~\cite{li2019traceable}. Here, we show that in a `table' data structure, it is straight-forward to deal with privacy preserving or full auditing, and to introduce auditing between confidential assets. Similarly, ZCash which is based on succinct schemes,  is also being studied for the propose of regulation enforcement protocols~\cite{xu2023regulation,li2019toward,xu2023regulation}, or for the ability to introduce multi-assets for ZCash ~\cite{Ding2020a,Xu2020}. These ledgers and others based on zk-SNARK \cite{Ding2020a,Jia2024,Xu2020,Garman2016,Wust2021}, however do not focus on enterprise blockchains use-cases where trust setup as with original zk-SNARK implementations is not practical for auditing context. 
Finally, there are also applications of privacy preserving ledgers that are in this work context, for example, \cite{khattak2020dynamic,meralli2020} for dynamic pricing and securities market, respectively. We extend auditing capabilities, and add to these set of applications by applying PADL to private atomic swap exchanges and bond markets in a single framework and general transaction scheme.

The rest of the paper is structured as follows. Sec. \ref{methods} describes the PADL design and the transaction scheme informally. This is immediately followed by the application use-cases in Sec. \ref{ch:usecases}. Sec. \ref{sec:padlformal} describes formal treatment of our PADL followed by its security and performance analysis in Sec. \ref{sec:eval}. For supplementary, we add appendix providing detailed examples of transactions (\ref{appendix:usecases}), and further security analysis with complementary proofs (\ref{appendix:security},D).

\section{PADL - Ledger Design}
\label{methods}

This section describes PADL design and the proposed private transaction scheme informally. We postpone the formal definition of PADL transaction scheme to Sec. \ref{sec:padlformal} but for clarity of this section, introduce essential components; Let $\mathbb{G}$ be a cyclic group and let $\mathrm{g}, \mathrm{h}$ be two random generators of $\mathbb{G}$. Then in the setup of a ledger, each participant generates its public key, $\pactpk$ which is related to its secret key $\pactsk$ via the following relation: $\pactpk = \mathrm{h}^{\pactsk}$.

\begin{figure}[h!]
  \begin{tabular}[c]{@{}c@{}}
    \begin{subfigure}[c]{.37\linewidth}
      \centering
      \includegraphics[width=1.1\linewidth]{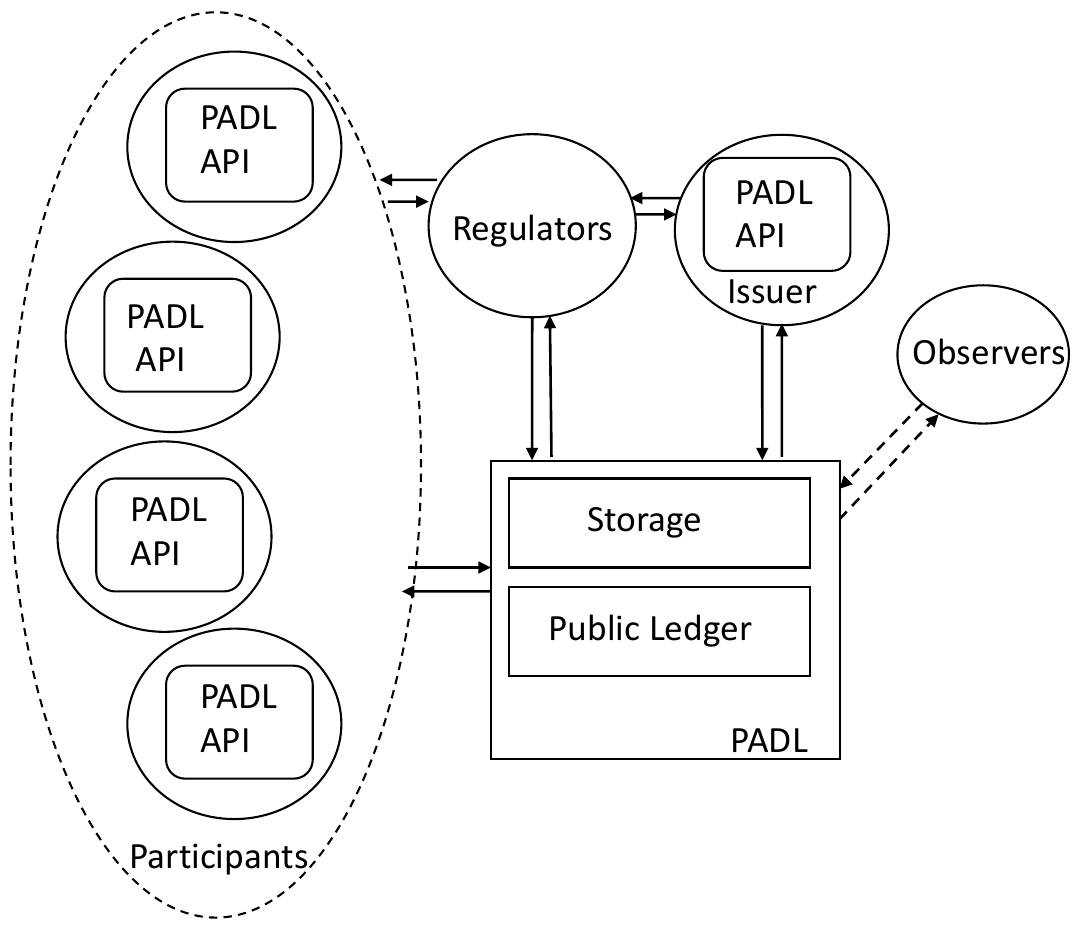}%
      \caption
        {%
          PADL ledger service setup.%
          \label{fig:upper}%
        }%
    \end{subfigure}\\
    \noalign{\bigskip}%
    \begin{subfigure}[c]{.4\linewidth}
      \centering
      \includegraphics[width=1.1\linewidth,page=2]{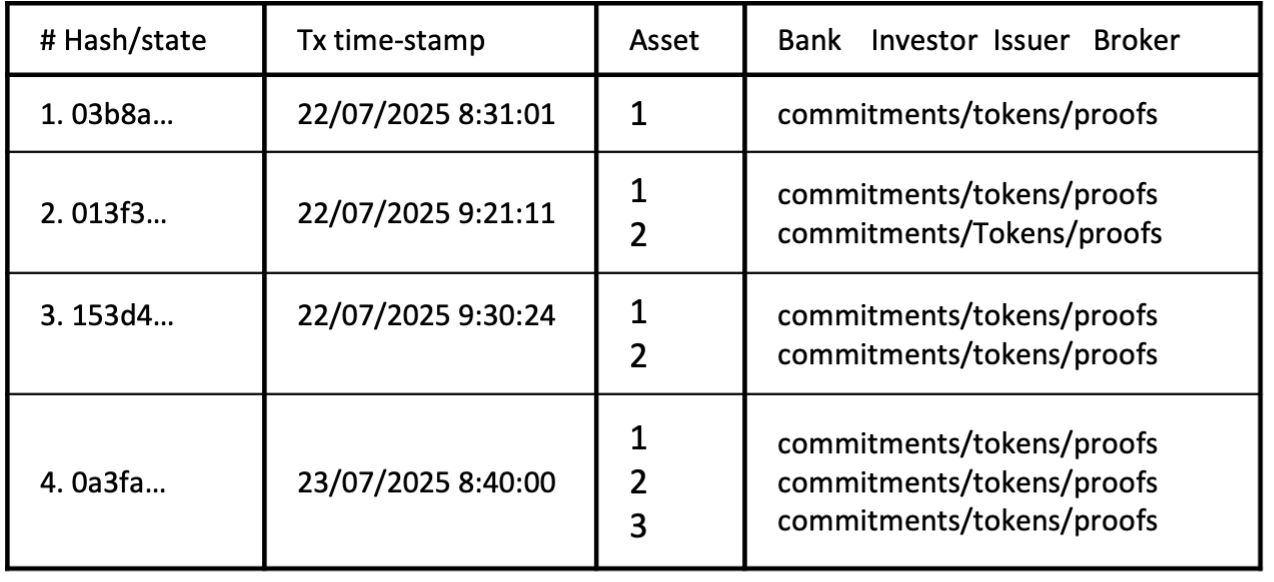}%
      \caption
        {%
          3D 'table' with indices of transactions {$t$} participants ($p$) and assets ($a$).%
          \label{fig:lower}%
        }%
    \end{subfigure}
  \end{tabular}
    \begin{subfigure}[c]{.60\linewidth}
    \centering
    \includegraphics[width=0.9\linewidth]{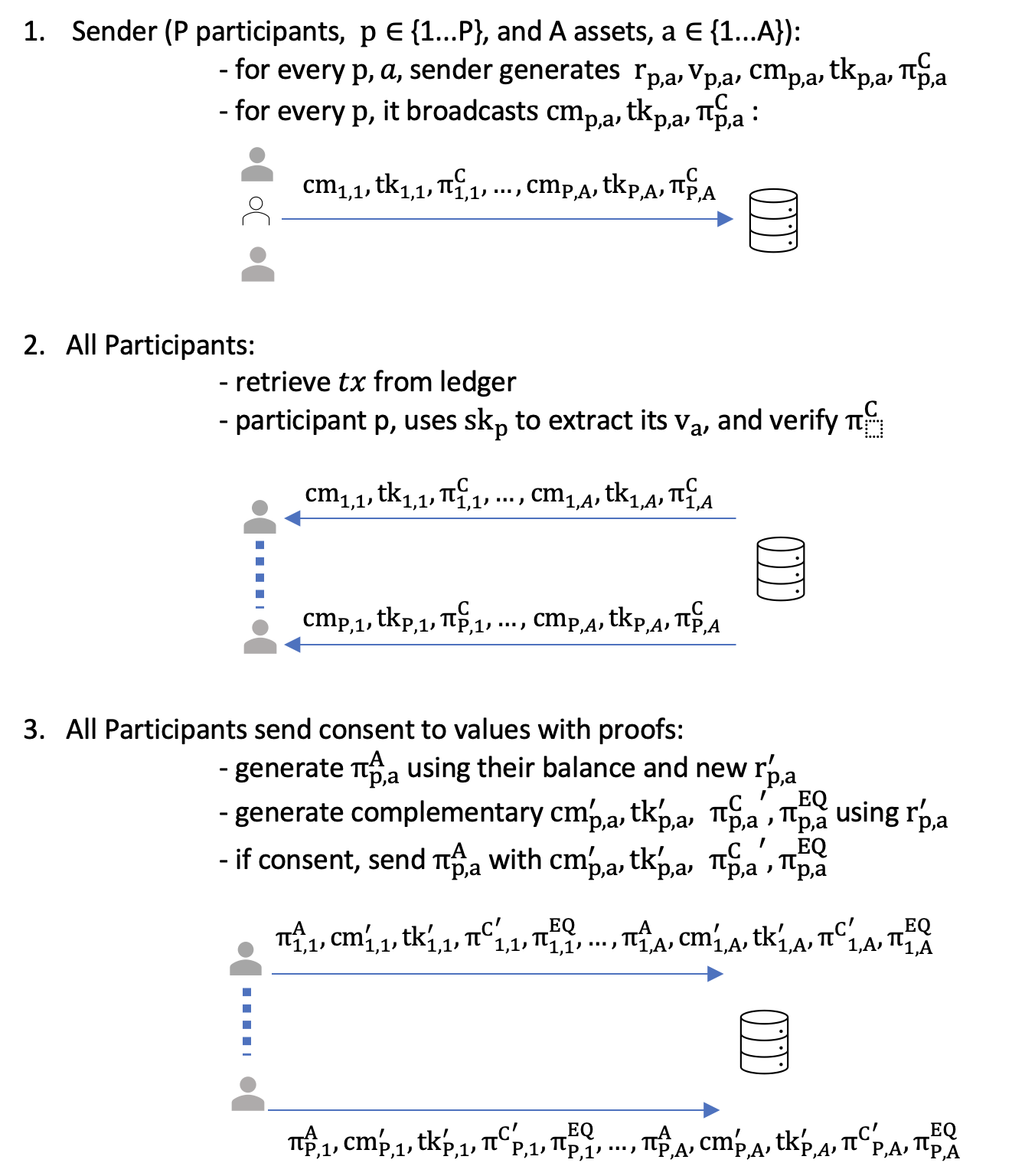}%
    \caption
      {%
        Transaction scheme and communication.%
        \label{fig:big}%
      }%
  \end{subfigure}\hfill
  \caption
    {%
      PADL high-level design of transaction and a `table' ledger.%
      \label{fig:every}%
    }
\end{figure}
The PADL structure is an `append-only' transaction ledger, where each transaction is added to a table in a synchronized manner. Figure \ref{fig:every} illustrates the different actors that can interact with the distributed ledger (a), the information stored in the encrypted ledger table (b) and high-level description of the transaction scheme (c). Asynchronous transactions is also possible in various ways, but out of scope at this stage.
A transaction is sent to the host of the ledger by the `sender'. The sender knows and hides all participants' values for a specific transaction. A multi-asset transaction, consisting of one or more assets, represents one row that may have several sub-rows depending on the number of assets in the transaction. Each transaction can possess multiple confidential assets, so the assets and their amounts in the transactions are encrypted using Pedersen commitment for each participant to hide their values. The commitments are accompanied with tokens and the proofs required to validate the integrity of the transaction scheme.
Thus, for each transaction $t$, each participant $p$, and each asset $a$,  we associate a cell of a 3-dimensional and dynamical array, $\cell_{t,p,a}$, where $t,p,a$ are the indices of the transaction, participant, and asset, respectively: 

\begin{equation*}
\cell_{t,p,a} = \{\commitment_{t,p,a}, \token_{t,p,a}, \commitment'_{t,p,a}, \token'_{t,p,a}, \pactproofs^{\asset}_{t,p,a}, \pactproofs^{\consistency}_{t,p,a}, \pactproofs^{\consistency'}_{t,p,a}, \pactproofs^{\equivalence}_{t,p,a} \} 
\end{equation*}

In the above $\cell_{t,p,a}$, $\commitment_{t,p,a}$ is Pedersen commitment, defined by 
$\commitment_{t,p,a} = \mathrm{g}^{\pactcoinvalue_{t,p,a}} \mathrm{h}^{\coinrand_{t,p,a}}$. Here $\pactcoinvalue_{t,p,a}$ is the value participant wishes to hide and $\coinrand_{t,p,a}$ is a random blinding factor. The commitment is associated with a blinding factor as a token, $\token_{t,p,a}=(\pactpk_p)^{\coinrand_{t,p,a}}$.  In order to prove the binding of the blinding factor in the token to the commitment, a proof of `consistency' is always accompanied here ($\pactproofs^C_{t,p,a}$). To make sure that no asset is created or destroyed in a transaction, it is required that the product of commitments across participants for each asset is equal to 1, where $\forall\ a,t, \Pi_{p}\commitment_{t,p,a} = 1 \ $. This ensures that the sum of blinding factors and the sum of values are equal to zero, $\sum_{p}\pactcoinvalue_{t,p,a} = 0$ and $\forall \ a,t, \sum_{p}\coinrand_{t,p,a} = 0 \ $. 
Table \ref{table:padl_algo} provides the algorithms for PADL construction, and the generation of commits and blinding factors is detailed in the algorithm $\pactspend$.

Additionally, a proof of asset is required to validate that no participant overspends, $\pi^A_{t,p,a}$. In order for each participant to provide a proof of asset, we introduce a complementary commit $\commitment'_{t,p,a}$, a complementary token $\token'_{t,p,a}$, and proofs. The complimentary commit is a commitment to the same value but with a different blinding factor $\coinrand'$ and the token is the public key raised to $\coinrand'$. The complimentary commits are required since the participants do not know the blinding factor in the original commitment, hence they re-commit their values. Finally, each participant starts with an initial cell for each asset that can be minted by the participants themselves and/or verified by other parties (Def. \ref{definition:pact}).

\subsubsection{PADL Transaction.}
\label{txscheme}
The transaction scheme is designed for general purpose, confidential multi-assets transactions (Fig.~\ref{fig:every}c). A formal definition of the transaction scheme (Def. \ref{definition:pact}) is provided later at Sec. \ref{sec:padlformal}. Before the broadcasting step, a sender generates commitment-token pairs and proofs of consistency, for every asset and participant. It also hashes the commitment-token pairs $H(\commitment,\token)$ to provide a unique identification of the transaction information. The transaction is then broadcast, without the proof of assets, allowing negative values for other participants, e.g. a broker, or in `atomic' asset exchange. Subsequently, it also does not reveal who is the sender (traditionally known as spender), i.e. the sender is not the only participant who can commit to negative value (sends asset). After broadcasting the transaction, each participant, can extract the values committed, $\extract(\commitment_{t,p,a},\token_{t,p,a},\pactsk_{p})$ using their secret keys by calculating $\commitment_{t,p,a}^{\pactsk_p}\cdot \token_{t,p,a}^{-1}=g^{{\pactsk_p}^{\pactcoinvalue_{t,p,a}}}$ and use brute-force to solve the $\mathrm{dlog_{g^{\pactsk_p}} g^{{\pactsk_p}^{\pactcoinvalue_{t,p,a}}}}=\pactcoinvalue_{t,p,a}$. Only the holder of the secret key can extract and be convinced with its value. This extraction is possible in practical cases where value $\pactcoinvalue_{t,p,a}$ belongs to a finite small set, and within a maximum range.

Next, the participants fill their proof of assets. Since $\coinrand_{t,p,a}$ is not known for the participant, participants provide also a complementary commit($\commitment'_{t,p,a}$), complementary token ($\token'_{t,p,a}$) and complementary consistency proof ($\pactproofs^{\consistency'}_{t,p,a}$) for a new private $\coinrand'_{t,p,a}$.  In order to verify that the same $\pactcoinvalue$ is used, the participant also generates an equivalence proof ($\pactproofs^{\equivalence}$). This proof is also sufficient to verify the identity, as it requires knowledge of its secret key, this proof is detailed in Sec. \ref{sec:padlformal}. The participant then adds a proof of asset $\pi^A$ for proving the sufficient balance in asset $a$ with the product of previous commits, $\prod_{t} \commitment'_{t,p,a}$. 

This scheme has several advantages. First, the proving system is simpler than previously suggested~\cite{narula2018zkLedger,chatzigiannis2021miniledger}, since it does not require generating the conditions of using a disjunctive proof.

Second and more importantly, it enables a single exchange with multi-assets, by committing minus values for other participants in a transaction. This offers practical flexibility when applying PADL to different use-cases where a trade-off between privacy, performance and future commitments/contracts is needed. Third, it hides the sender from the broadcaster or from the ledger. Fourth, generating the proof of asset is computationally expensive and likely the limiting step, and in this scheme, the participants generate the proof for themselves, asynchronously.
\subsubsection{Extension to `injective tx'.}
PADL transaction scheme requires other participants besides the sender to interact with the ledger to fill up the proof of assets. However, this is not necessarily a bottleneck since on any encrypted ledger, participants would need to consent to an encrypted value, even if it is a positive value, e.g. in loan assets, or where assets accrued interest. With that, our scheme can be easily changed to the simple case of an `injective' transaction where consent is not required. In such case, the sender would provide range proofs to all transaction values to be positive, and proof of asset for itself. Along with another disjunctive proofs for everyone to conceal the sender~\cite{narula2018zkLedger}. Note that in any case the validation of proofs is required.

\subsubsection{Validation and extension to dropping parties.}
In order to validate transactions, each participant can verify the ZKP proofs at its own local node or machine, and also validate the hash declared in the transaction to ensure the commits and tokens corresponds to the number of participants and assets cells seen. An API can approve or reject the transaction. Alternatively, verification of the proofs for all broadcast cells can be obtained by a centralized service or with smart contract (which is also implemented in the PADL code). 
In case that some proof of assets are not valid, or not filled due to dropping parties,
the sender can exclude these participants ($O\subset P$) and the transaction can be still validated. The exclusion is simply done by an extra step. The sender($p_s$), regenerates its own cell with a modified blinding factor, $\coinrand_{p_{s}}$, that is the sum of the excluded participants and itself, $\coinrand_{t,p_{s},a}=\coinrand^{org}_{t,b,a}+\sum_{o \in O}{\coinrand_{t,o,a}}$.
That makes the proof of balance valid again. In fact, the sender is able to remove any participant as long as the commit value is 0. Otherwise, the proof of balance would fail. The trade-off here, is anonymity of the transaction graph. The less participants involved in the transaction, the more information on the transaction graph can be extracted. 
The security of anonymity and integrity of the transaction scheme is defined in Sec. \ref{sec:eval} and formally analyzed in Appendix \ref{appendix:security}.

\section{Use-Cases}
\label{ch:usecases}
We discuss several fundamental use cases for transactions and trading applications where PADL can provide privacy with audit-ability. Other uses cases can be built on top of the ledger components here. Concrete examples of transaction generation are also shown in Appendix \ref{appendix:usecases}.  
\subsection{Simple exchange ledger}
\label{simple_exchange_ledger}
We describe how PADL can be seamlessly used for exchange of confidential assets atomically. Assuming that two participants would like to exchange assets, with PADL it can be done in a single transaction. To do so, participant A creates a transaction consisting of commitments and tokens committing to transferring $x$ amount of asset $a$ to participant B. In the same transaction, participant A also writes a commitment on behalf of participant B, committing participant B to $y$ amount of $b$. Then, participant A can broadcast this transaction, participant B can decide to either accept or reject this transaction. 
To accept the transaction, participant B simply writes its Cell's proofs, which is then finalized and verified by participant A and the transaction is appended to the ledger. Note that by adding their proof of assets, they consent to the exchange. Next, the transaction is appended if validation passes. In the validation, the ZKPs are verified and the immutability of original broadcast transaction is checked by checking the hash of the final transaction. Finally, if a participants provides a `wrong' or `bad' cell in the transaction, it simply results in rejection of the transaction. 

\subsection{Settlement Trusted Bank ledger}
\label{settlement_ledger_trusted}
Currently, bank payment systems are required a trusted party in order to be audit-able and follows `Know-Your-Costumer', Anti-Money Laundering, specific rules, or keys recovery criteria. In such cases, the ledger contains a party that acts as a settlement bank, and would have access to the balance of some participants ('full' auditing). In a distributed solution, we wish to avoid key management by the trusted parties, meaning parties would not share their private keys with any party. This scenario can be treated in PADL, by an additional token, $\token^{\textit{I}}$; signing $\coinrand_{t,p,a}$ by the public key of the trusted party. For many cases, the issuer can be defined as the `auditor' or `trusted' party,  where its participant index is $p=I$. The transaction cell structure is altered to facilitate this requirement:
\begin{equation*}
\cell_{t,p,a} = \{\commitment_{t,p,a}, \token_{t,p,a}, \commitment'_{t,p,a}, \token'_{t,p,a}, \pactproofs^{\asset}_{t,p,a}, \pactproofs^{\consistency}_{t,p,a}, \pactproofs^{\consistency'}_{t,p,a}, \pactproofs^{\equivalence}_{t,p,a}, \token^\textit{I}_{t,p,a},  \pactproofs^{\consistency^\textit{I}}_{t,p,a} \} 
\end{equation*}
Note that we also need to add consistency proof for the new added token. However, if in the committee consensus, the issuer is also designated as the approving entity (no need for consensus), the Cell structure can be reduced to:
\begin{equation*}
\cell_{t,p,a} = \{\commitment_{t,p,a}, \token_{t,p,a}, \pactproofs^{\consistency}_{t,p,a}, \token^\textit{I}_{t,p,a},  \pactproofs^{\consistency^\textit{I}}_{t,p,a} \} 
\end{equation*}
The issuer can also check for any asset balance, by extracting the balance for each asset itself, and proof of asset can be avoided. This accelerates largely the ZKP system for a very typical and realistic use-case. In addition since, the token is designated to a single cell, it is flexible to define many audit tokens for a list of parties ($\pactpk$s) and for any cell in a transaction. 

\subsection{Bond Market ledger}
\label{bond_market_ledger}
In this subsection we describe the application of PADL to a third and a more complex use case of trading bonds. The application of blockchain in the debt capital markets with arrival of `smart bonds' is highly anticipated. However, in many cases, the bond issuers or investors might be reluctant to disclose their positions. We show how PADL is applied to make private transaction in the bond markets. For this case, we are mainly interested in exchange of two types of assets, bonds and a USD-backed digital token. The bond issuer would like to issue bonds and borrow USD-coin in return. The investors wish to lend USD-coin to buy bonds, however, the USD-coin for the investors are handled by another party, i.e. custodians. Also, the bonds are handled by brokers with fees, thereby interacting with the bond issuer and the investors. 
The privacy requirements in this scenario are as follows: 
\begin{itemize}
    \item Only the broker knows the details of bond deals.
    \item Custodian knows only the amount of money released to an investor.
    \item The bond issuer does not need to know the individual investors' contribution.
    \item The investors does not know about other investors.
    \item Coupon (or interest rate) payment by the issuer can be issued without the issuer learning about the distribution of the payment amongst the investors.
\end{itemize}
PADL allows the investors and bond issuer to make these transaction while maintaining privacy as described below (in brief). The entire transaction flow is described in Appendix. 

We start by assuming all the participants: issuers, custodians, broker and the bond issuer are on PADL, and a ledger consists of two assets, a bond and a USD tokens. The ledger is initialized such that bonds are minted by the bond issuer and the USD coin asset by the custodians. 
The custodian can then issue USD token to the investors as requested by the investors, by creating a transaction on PADL. In order to exchange assets (buy bond for USD coin), the broker creates a transaction on behalf of both the investor and bond issuer. This transaction consists of commitments that are written in every participant's columns. The bond issuer can open the commitments in  its column to verify that it is committing to $x$ number of bonds in exchange for $y$ number of USD coins. Similarly, the investor can verify that it is committing to USD token in exchange for receiving bonds. Furthermore, the broker creates transactions for the bond issuer to pay the coupons to the investors. Again, the only value the bond issuer can verify is that it pays the correct coupon rate for the money it received for the bonds, however, it cannot decrypt the payments that each of the participants on PADL receive. Each investor however can verify that it receives the right amount of coupon payment as promised. 
Finally once the bond matures, the broker in a single transaction (see example in Appendix \ref{appendix:usecases}), can return the USD coins to the investors while taking its commission away from all participants, and the bond issuer can destroy the bond asset.

The use-cases above shows the applicability of PADL to develop various private and audit-able markets. We next discuss the cryptography and methods which define the transaction scheme, auditing, and ledger construction followed by security analysis.
\section{Methods}
\label{sec:padlformal}
\subsection{Notations and Background}

We denote by $y \leftarrow\adv(x)$ the execution of algorithm $\adv$ on input $x$ and produces $y$. By $\adv^\oracle$, we denote an algorithm $\adv$ that has access to oracle $\oracle$. We use $r \sample S $ to denote that r is sampled uniformly at random from a set S. 
Let $\group$ be $p$ prime-ordered additive subgroup of the elliptic curve $E(\mathbb{Z}_q)$ generated by some generator $g$ and some prime $q$. For simplicity, we use multiplicative notation for group $\group$ throughout the paper. 
By $\commitment_{t, p, a}$, we denote the participant $p$ (with public key $\pactpk_p$) account's commitment coin in the $t$-th transaction row for a specific asset $a$. Similarly, by $\pactcoinvalue_{t, p, a}$, we denote the corresponding commitment coin value. We exploit the notation for 
% $\pacttx_{\pacttx, \pactpk, a}$  and
$\pactcoinvalue_{\pacttx, p, a}$ to mean the same for a specific transaction $\pacttx$. We often omit $a$ if the context is clear. \\
% We use Dictionary[Key] to store key-value pair mapping that can be looked up in constant time. 
% \YS{what is pk's coin? commitment? WZ: changed wording}
\noindent\textbf{Discrete Log (DLOG) Assumption:} Given $g, g^x$ where $g \in \group$ and $x \in \grouporder$, no PPT adversary can output $x$ with non-negligible probability. This paper relies on the intractability of solving the discrete logarithm problem.

\subsubsection{Commitment Scheme.} A commitment scheme consists of two steps: first the sender commits to a value as a commitment; later the sender may choose to reveal the committed value. A commitment scheme should satisfy the hiding and binding properties. Additionally, a commitment scheme has additive homomorphic property, if given two commitments,
% , it is possible to compute a new commitment with the summed 
% message and randomness which means for all well-formed commitment key $\commitkey$, commitment $\commitment_0, \commitment_1$, and randomness $\coinrand_0,\coinrand_1$, 
we have $\commit(\commitkey, m_0, \coinrand_0) \cdot \commit(\commitkey, m_1, \coinrand_1) = \commit(\commitkey,m_0+m_1 , \coinrand_0+\coinrand_1)$. 
We briefly recall notions for the Pedersen commitment scheme \cite{C:Pedersen91} while hiding and binding definitions \cite{EC:GroKoh15} are recalled in the Appendix. 
% In addition, a commitment scheme is said to be (publicly verifiable designated) extractable if it is publicly verifiable that having a corresponding secret key will enable extraction of the commited message without interaction and the extractable party could be chosen based on some public key.
$\commitkey$ is omitted for the rest of the paper since it is clear from the context.

\begin{definition}[Pedersen Commitment]
\label{comdef}
Pedersen (homomorphic) commitment scheme $\sf{(HCOM)}$ consists of the following set of algorithms:
	\begin{itemize}
	\item{ $\commitgen\textnormal{(}\secparam\textnormal{)}$:}  on input a security parameter $\secparam$, this algorithm computes $\ \commitkey :=(g,h) \sample \group$ where $g,h$ are generators and outputs $\commitkey$.
    \item{$\commit\textnormal{(}\commitkey, \msg, \coinrand \textnormal{)}$:} on input a commitment key $\commitkey$, a message $\msg$ and a randomness $\coinrand$, this algorithm parses $(g,h) := \commitkey$
    % , sample $\coinrand \sample \grouporder$ 
    , and outputs $g^{\msg}h^{\coinrand}$ as the commitment $\commitment$.
	\end{itemize}
	\label{definition:com}
\end{definition}

\subsubsection{Zero-Knowledge Proof.} A non-interactive zero-knowledge proof (NIZK) \cite{EC:GroOstSah06} allows a prover to prove to a verifier of some statement in zero-knowledge. NIZK is defined as follows:

\begin{definition}[Non-interactive Zero-Knowledge Argument System] Let $\relation$ be an NP-relation and $\lang_{\relation}$ be the language defined by $\relation$. A non-interactive zero-knowledge argument system for $\lang_{\relation}$ consists of the following algorithms:
    \begin{description}
    \item[$\nizksetup_{\lang_{\relation}}(\secpar)$:] Takes as input a security parameter $\secpar$, and outputs a common reference string $\crs$.
    \item[$\nizkprove_{\lang_{\relation}}(\crs, \stmt, \wit)$:] Takes as input a common reference string $\crs$, a statement $\stmt$ and a witness $\wit$, and  outputs either a proof $\nizkproof$ or $\ground$.
    \item[$\nizkverify_{\lang_{\relation}}(\crs, \stmt, \nizkproof)$:] Takes as input a common reference string $\crs$,  a statement $\stmt$, and a proof $\nizkproof$,  and outputs either a $0$ or $1$.
    \end{description}
	\label{definition:nizk}
\end{definition}
\label{txscheme}

\subsection{ Privacy Preserving Multi-Asset Transaction with Audit}

In this section, we formally introduce the notion of a confidential transaction with privacy-preserving audit scheme ($\pact$) which is later used to show PADL construction. $\pact$ is based on the syntax from (ring) confidential transaction (Ring-CT) \cite{ESORICS:SALY17,FC:YSLAEZG20} but additionally introduces notions that capture the requirements for a variety of privacy-preserving audits. Modifications are also made to the syntax of confidential transaction scheme, allowing the scheme to better capture the functionality of table-based multi-asset transaction rather than graph-based mono-asset transaction. The security and privacy of the transaction scheme are detailed in Appendix \ref{appendix:security}.

\subsubsection{$\pact$ scheme.}
In general, a $\pact$ scheme consists of algorithms that are used to make confidential transactions. 
% We additionally define a set of oracles for interacting (to post and retrieve transactions) with the table-based distributed ledger that is using $\pact$ transaction scheme. 
It is assumed that the distributed ledger is properly maintained and updated at all time\footnote{A consensus protocol is used to manage the state of the distributed ledger. However, the exact details are considered outside the scope of this work.}. Recipient accounts (accounts where the commitment coin is positive) and spending accounts (accounts where the commitment coin is negative) are implicitly captured by the transaction amount list $\pactvalue$. 
% Scheme specific information such as spending accounts secret key list $\pactssk$ is implicitly captured by auxilliary information $\pactaux$. 
Note that the transaction considered here is a multi-asset transaction. We additionally introduce $\pactconsent$ to capture the transaction flexibility of the proposed PADL scheme, which permits policy-based spending. Policy enforcement is flexible and can include rules such as always accepting positive deposits, only accepting negative amounts if the transaction is initiated by the account owner, only accepting negative amounts if it is an authorized scheduled direct debit, and other possible variations. In practice, the default policy, $\pactpolicy_{G}$ is configured to accept positive amounts and only accept negative amounts if the transaction is initiated by the same account.

\begin{definition}[$\pact$]
A privacy preserving-enabled confidential transaction \\ scheme $\pact$ \text{is a tuple of algorithms}, \\ $\pact =(\pactsetup,\pactkeygen,\allowbreak\pactmint\allowbreak,\allowbreak\pactspend,\pactverify)$: %
	\begin{itemize}
	\item{$\pactsetup\leftbrc\secparam\rightbrc$:} on input a security parameter $\secparam$, this algorithm outputs a public parameter $\publicpar$. All algorithms defined will implicitly receive $\publicpar$ as part of their inputs.
    \item{$\pactkeygen\leftbrc\secparam\rightbrc$:} on input a security parameter $\secparam$, this algorithm outputs an account secret key $\pactsk$ and an account public key $\pactpk$.
    \item{$\pactmint\leftbrc \pactcoinvalue \rightbrc$:} on input a transaction amount $\pactcoinvalue$, this algorithm outputs a commitment $\commitment$ and a commitment blinding factor $\coinrand$.
     \item{$\pactconsent\leftbrc \pactcoinvalue, \pactpolicy, \pactsk, \pactaux \rightbrc$:} on input a transaction amount $\pactcoinvalue$, a transaction policy $\pactpolicy$, an account secret key $\pactsk$, and an auxiliary information $\pactaux$, this algorithm outputs an endorsement data $\dataend$.
     \item{$\pactspend\leftbrc \pactvalue, \pacttxs, \pactpks \rightbrc$:} on input a transaction amount list $\pactvalue$, a history transactions list $\pacttxs$, and a master account public key list $\pactpks$, this algorithm outputs a new transaction $\pacttx$, and a validity proof $\pactproofs$.
    \item{$\pactverify\leftbrc \pacttx, \pactproofs,\pacttxs,\pactpks \rightbrc$:} on input a new transaction $\pacttx$, a validity proof $\pactproofs$, a history transactions list $\pacttxs$, and an account public key list $\pactpks$, this algorithm outputs a verification bit $b \in \{0,1\}$.
	\end{itemize}
	\label{definition:pact}
\end{definition}

%\YS{Wouldnt we need a definition for tx and pi above? }
%\YS{spending account secret key list??}
%\wz{TODO: tx and pi can be redefine properly since definition is now moved after design.}

In the context of auditing distributed ledgers, it is often required that the auditor can open the transactions and trace their graph to make conclusions about the financial health of the organization. In Provisions \cite{CCS:DBBCB15}, more privacy-preserving auditing is considered whereby solvency is proved without revealing the entire transactions history. However, existing works do not consider a couple of fundamental audit information that is used in financial auditing landscape. We identified the following sets of privacy-preserving auditing concepts that would be of interest to the financial auditing of participants in the distributed ledger. $\pact$ should support the following set of privacy-preserving audit capabilities:
\begin{enumerate}
    \item \textbf{Basic Asset Balance. } Without revealing the transaction history, the verifier should be convinced of the asset balance of the ledger participant.
    \item \textbf{Liquidity. } Without revealing the transaction history and an asset balance, the verifier should be convinced of the liquidity or credit of the asset.
    \item \textbf{Inter-transactions Rate. } Without revealing the transaction history, the verifier should be convinced of the inter-transaction rate between two subset of transactions.
\end{enumerate}

% \begin{definition}[Privacy-preserving Audits] 
%     \begin{itemize}
%         \item \textbf{Basic Asset Balance. } Without revealing the transaction history, the verifier should be convinced of the asset balance of the ledger participant.
%         \item \textbf{Liquidity. } Without revealing the transaction history and an asset balance, the verifier should be convinced of the liquidity or credit of the asset.
%         \item \textbf{Inter-transactions Rate. } Without revealing the transaction history, the verifier should be convinced of the inter-transaction rate between two subset of transactions.
%     \end{itemize}
%     \label{definition:audits}
% \end{definition}

\subsection{Final PADL construction}
\label{sec:padlcon}

In this section, we lay out the cryptography setup and elements involved in the transaction scheme.
As discussed in the preceding sections, $\pactsetup$ is executed to obtain $g,h$ which are assumed to be public in the ledger. Each participant will generate their keypair $\pactsk_p,\pactpk_p$ using $\pactkeygen$ where $\pactpk_p$ will now be the participant's identity and unique account address. 
Transactions with proofs are generated using $\pactspend$ while the transaction value list $\pactvalue$ used could be prepared by any party.
During the generation of a transaction, the sender contacts the broadcaster using $\broadcastfunction$ (in PADL, broadcast is a functionality provided by the ledger) to obtains endorsements from all parties. Meanwhile, the transaction amount is committed using $\pactmint$. Lastly, $\pactverify$ is used to verify the transaction with proofs generated by $\pactspend$ with respect to the latest ledger transaction state and participant list. Each participant maintains its internal state $\stat_p$ which records $\pactsk_p, \pactpk_p, \pactpolicy_p$. 
The construction is summarised in Table \ref{table:padl_algo} and all the elements of PADL's $\cell_{t,p,a}$ are listed as follows (indices $t$, $p$ and $a$ are omitted whenever appropriate for fluency).

% where  the specific language?
% Specifically, in PADL setup, $g$ and $h$ are  two point generators of the cyclic group $G$, here, from the elliptic curve parameters over the finite field $\mathbb{F}_p$ and with the size $|G|=q$. 
% $\pactpk_p$ is generated by each participant, by choosing a random $x$-point on an elliptic curve in the group $\mathbb{G}$ to be its secret key, $\pactsk_p$. Then obtain:
% \begin{equation*}
% \pactpk_p = h^{\pactsk_p}
% \end{equation*}

% It is required in a Pedersen scheme that nobody knows $\log_gh$, to ensure that a committer cannot open commitment to a different value $v'$, but here, $r$ is also signed with an authentication token. So 

% The participant can then validate $v$, or extract $v$ if it is sufficiently small with a generic discrete logarithm algorithm, i.e. \cite{pollard1978monte}. Extracting $v$, enables committing a value for a participant without interaction with the participant. In the case of a transaction, the broadcaster generates commitments for all participants and blinding factors are generated by the broadcaster to follow:   
% \begin{equation}  \label{eq:zerosumr}
% \sum_{p=0}^{N_p}r_{a,p}=0,\\ \forall a \in \{0,1, \dots, N_a\}
% \end{equation}

% \wz{add in how to generate $\pactvalue$ and then use in $\pactspend$ to generate $\pacttx$ and $\pactproofs$, then during $\pactspend$, $\pactconsent$ is invoked, and say how broadcaster $\broadcastfunction$ called $\fillasset$. In addition, we assume each participant have its internal state $\stat_p$ that records $\pactsk_p, \pactpk_p, \pactpolicy_p$ }

\noindent\textbf{Commitment} commits to the value in a transaction as defined in Def. \ref{comdef}: $\commitment := \commit(\pactcoinvalue,\coinrand)$.

% To hide a value $\pactcoinvalue$, the committer chooses a random point $\coinrand$, and computes the commitment as
% \begin{equation*}
% \commitment(\pactcoinvalue,\coinrand) =g^{\pactcoinvalue}h^{\coinrand}
% \end{equation*}

% where ${r}\in [0, \dots, q-1]$, $v$ is in a finite small set by the ledger and $q$ is the number of elements in $G$, $q=|G|$. The commitment scheme alone is perfectly hiding, as for every $v'\neq v$ there exist $r,r'\in [0..q-1]$ such that $cm(v,r)=cm(v',r')$. 

% Schnorr signature~\cite{C:Schnorr89} 
\noindent\textbf{Token} hides the blinding factor, $r$ with a public-key,
$\token(\coinrand,\pactpk)=\pactpk^{\coinrand}$.
The token has two distinctive roles: to facilitate verification of proof and to enable extractability of the committed value without knowledge of $\coinrand$, as long as the value belongs to a small set. For example, participant, $p$, can take a commitment $\commitment(\pactcoinvalue,\coinrand)$ and token $\token(\coinrand, \pactpk)$, and use its secret key to compute $\frac{\commitment^{\pactsk}}{\token}$ to obtain $({g^{\pactsk}})^{\pactcoinvalue}$. The extractability is also used in assigning an auditor or a trusted party, as described in the settlement bank use-case.  
Note that the combined system ($\commitment$, $\token$) makes the scheme computationally hiding because adversary could solve the Discrete Log Problem for $\token$ to obtain $\coinrand$ and be able to open the $\commitment$ and proof of consistency ensured that the correct $\coinrand$ is used to construct the token.

% Note that if an adversary tries to open $\commitment$ to a different value by finding $\commitment(\pactcoinvalue',\coinrand')=\commitment(\pactcoinvalue,\coinrand)$, it would also require $\pactpk^{\coinrand}=\pactpk^{\coinrand'}$. This means that even if a committer knows $\log_gh=x$, and chooses $\coinrand'$ such as $\coinrand'=\frac{\pactcoinvalue-\pactcoinvalue'}{\textit{x}}+\coinrand$, if $\pactcoinvalue' \neq \pactcoinvalue$, then with a negligible probability, $\pactpk^{\coinrand}=\pactpk^{\coinrand'}$. 
% Hence the combined system does not require a trust on choosing $g,h$. 

% \newline

\noindent\textbf{Proof of balance}maintains the conservation of all assets in a transaction $\pacttx$, ensuring that the sum of each asset value in the transaction equals zero. Proof of balance is a proof for the language $\lang_{\relation_{\balance}} := \{(\commitment_{\pacttx,1},\commitment_{\pacttx,2},\cdots, \commitment_{\pacttx,|\pactpks|}) : \sum_{p=1}^{|\pactpks|}{\pactcoinvalue_{\pacttx,p}}=0\}$. In PADL, proof of balance is checked by multiplying the commitments together and checked for equality with $1$. By the homomorphic property of Pedersen commitment and assuming $\coinrand_i$ is the additive sharing of zero, we have  $\Pi_{p=1}^{|\pactpks|} \commitment_{\pacttx,p} = g^{\pactcoinvalue_{\pacttx,1} + \pactcoinvalue_{\pacttx,2} + \cdots + \pactcoinvalue_{\pacttx,|\pactpks|}}h^{\coinrand_{\pacttx,1} +\coinrand_{\pacttx} + \cdots + \coinrand_{\pacttx,|\pactpks|}} = g^0h^0 = 1$. The proof of balance is unforgeable due to the hardness of discrete log.
\newline 

\noindent\textbf{Proof of consistency} ($\pactproofs^{\consistency}$) ensures that the randomness used in the commitment and token are consistent. The proof of consistency is a proof for the language $\lang_{\relation_{\consistency}} := \{(\commitment,\token) : \exists \pactcoinvalue,\coinrand\textnormal{ s.t. }\commitment = \textit{g}^{\pactcoinvalue}\textit{h}^{\coinrand} \wedge \token = \pactpk^{\coinrand} \}$. 
The proof of consistency is a ZK-$\Sigma$-protocol where the \textit{prover} chooses random $u_1, u_2 \in \grouporder$, and sets $t_1 \leftarrow g^{u_1}h^{u_2}, \  t_2 \allowdisplaybreaks\leftarrow \pactpk^{u_2}$, and generates a challenge with a hash function $c=H(h,g, t_1,t_2, \commitment, \token) \in \grouporder$. The \textit{prover} calculates $\quad s_1 \leftarrow u_1+c\pactcoinvalue,\quad s_2 \leftarrow u_2+c\coinrand$. 
It then outputs the transcript $(t_1, t_2, s_1, s_2)$. 
The \textit{verifier} verifies: $\{t_1\cdot \commitment^c=g^{s_1}h^{s_2} \wedge t_2\cdot \token^c=h^{s_2} \wedge c=H(h,g, t_1,t_2, \commitment, \token)\}$. 
\newline

\noindent\textbf{Proof of equivalence} ($\pactproofs^{\equivalence}$) is used to show that two commitments have the same commitment value. The proof of equivalence is a proof for the language $\lang_{\relation_{\equivalence}} := \{(\commitment, \commitment') : \exists \pactcoinvalue,\coinrand,\coinrand' \textnormal{ s.t. } \commitment = g^{\pactcoinvalue}h^{\coinrand} \wedge \commitment' = g^{\pactcoinvalue}h^{\coinrand'}\}$.
Given a commitment-token pair ($\commitment,\token$) and their consistency-verified complementary counterpart ($\commitment',\token'$), we can compute $a := \frac{\token}{\token'}=h^{{(\coinrand-\coinrand')}^{\pactsk}}$ and $b:= \frac{\commitment}{\commitment'}=h^{\coinrand-\coinrand'},\quad \mathrm{\,\, assuming} \,\, \pactcoinvalue=\pactcoinvalue'$. If $\pactcoinvalue \neq \pactcoinvalue'$, the prover would need to compute discrete log $x= dlog_{(g^{\pactcoinvalue-\pactcoinvalue'}b)}b^{\pactsk}$ where $b=h^{\coinrand-\coinrand'}$ as the discrete log between $g$ and $h$ is not known.  
Therefore, it is sufficient to check for the knowledge of secret key between the base $b$ and group element $a$ to check for the commitment coin value equivalency. In PADL, Schnorr protocol \cite{C:Schnorr89} for knowledge of discrete log is used. 
Let $\frac{\token}{\token'}=a$ and $\frac{\commitment}{\commitment'}=b$. 
% We note there exists a natural homomorphism $\psi: Z_p \rightarrow \mathbb{G}: \psi(sk) = b^{sk}$, 
% such that for all $b \in \mathbb{G}$, $\psi(sk_a + sk_b)=\psi(sk_a) \psi(sk_b)$.
The \textit{prover} chooses random $u \in \grouporder$, sets $t \leftarrow b^{u}$, and generates a challenge with a hash function $c=H(a,b,t) \in \grouporder$. The \textit{prover} calculates $\quad s \leftarrow u+c\cdot \pactsk \in \grouporder$. 
It then outputs the transcript $(c, t, s)$. 
The \textit{verifier} verifies: $\{t\cdot a^c= b^s \wedge c=H(a,b,t)\}$. 
\newline

\noindent\textbf{Proof of asset} ($\pactproofs^{\asset}$) asserts that after summing the commitment coin values of an asset for a particular account p, it holds that the balance of the account is positive. Positive for PADL is defined to be within some range $\positive$. 
Proof of asset is a proof for the language $\lang_{\relation_{\asset}} := \{(\commitment_{1,p},\commitment_{2,p},\cdots, \commitment_{|\pacttxs|+1,p}, \positive) : 
\sum_{i=1}^{|\pacttxs|+1}{\pactcoinvalue_i} < \positive \textnormal{ where }g^{\pactcoinvalue_i}h^{\coinrand_i} = \commitment_i \}$.
Proof of asset for some participant p is instantiated by first asserting that a complementary commitment $\commitment'$ has the same coin value as the aggregated commitment coin via proof of equivalence whereby the aggregated commitment coin is obtained by multiplying all the commitment coins $\commitment^{}_{\Pi}:= \Pi_{t=0}^{|\pacttxs|+1} \commitment_{t,p}$\footnote{Alternatively, by multiplying all the individual complementary commitment coins which is a re-committed commitment coin.}. Then, a NIZK range proof for Pedersen commitment such as Bulletproofs \cite{SP:BBBPWM18} is used to assert that the asset balance contained in the complementary commitment  $\commitment'$ is positive.
\newline
% and $\pactproofs^{\equivalence}$
Note that $\pactproofs^{\consistency}$ is instantiated directly from the generalized special honest-verifier proof of knowledge ($\zkpok$) $\Sigma$-protocol for the preimage of a group homomorphism by Maurer \cite{AFRICACRYPT:Maurer09}. The aforementioned theorem and the instantiation are shown in the Appendix \ref{appendix:security}. Additionally, interactive proofs are transformed into non-interactive zero-knowledge proofs (NIZK) using the Fiat-Shamir transform \cite{C:FiaSha86} in the random oracle model and the transaction identifier is a part of the statement for all proofs. 
\begin{table}[htb!]
    \scriptsize
    \centering
    % First table
    \begin{subfigure}[b]{\textwidth}
        \centering
        \newcolumntype{R}{>{\centering\arraybackslash$}X<{$}}
        \begin{tabularx}{\linewidth}{R | R | R}
            \makecell[l]{
            % PACT setup
            \underline{\boldsymbol{\pactsetup\leftbrc\secparam\rightbrc}} \\
            (g,h) \gets \commitgen(\secparam) \\
            \publicpar := (g,h) \\
            \ret\publicpar 
            \\
            \\
            }
            &
            % PACT Keygen
            \makecell[l]{
            \underline{\boldsymbol{\pactkeygen\leftbrc\secparam\rightbrc}} \\
            (g,h) \gets \publicpar \\
            \pactsk \sample \grouporder \\
            \pactpk: = h^{\pactsk} \\
            \ret \pactsk,\pactpk
            \\
            }  
            &
            % PACT MINT
            \makecell[l]{
            \underline{\boldsymbol{\pactmint\leftbrc \pactcoinvalue \rightbrc}} \\
            (g,h) \gets \publicpar \\
            \coinrand \sample \grouporder \\
            \commitment := \commit(\pactcoinvalue, \coinrand) = g^{\pactcoinvalue} h^{\coinrand} \\
            \ret \commitment,\coinrand
            \\
            }
            \\ \hline
        \end{tabularx}
    \end{subfigure}
    % Second table
    \begin{subfigure}[b]{\textwidth}
        \centering
        \newcolumntype{R}{>{\centering\arraybackslash$}X<{$}}
        \begin{tabularx}{\linewidth}{R | R }
            % PACT Spend
                \makecell[l]{
                    \underline{\boldsymbol{\pactspend\leftbrc \pactvalue, \pacttxs, \pactpks \rightbrc}} \\
                    (g,h) \gets \publicpar \\
                    \pacttx \gets [\ ],  \pactproofs \gets [\ ] \\
                    \textnormal{for each asset } \pactvalue_a :=[\pactcoinvalue_1,\cdots,\pactcoinvalue_{|\pactpks|}] \textnormal{ in } \pactvalue:\\
                    \tab \pacttx_a \gets [], \pactproofs_a \gets [] \\ 
                    \tab \coinrand_0 = 0 \\
                    \tab \textnormal{for each participant } p, \textnormal{its } \pactpk_p \textnormal{ and } \pactcoinvalue_p:\\
                    \tab\tab \textnormal{if not last participant } p: \\
                    \tab\tab\tab \commitment,\coinrand := \pactmint(\pactcoinvalue_{p}) \\
                    \tab\tab\tab \coinrand_{0} += \coinrand \\
                    \tab\tab \textnormal{else}: \\
                    \tab\tab\tab \coinrand := - \coinrand_0\\ 
                    \tab\tab\tab \commitment := g^{\pactcoinvalue}h^{\coinrand} \\
                    \tab\tab \token:= \pactpk_{p}^{\coinrand} \\
                    \tab\tab \pactproofs^{\consistency} := \nizkprove_{\lang_{\relation_{\consistency}}}(\crs, (\commitment,\token), (\pactcoinvalue_p,\coinrand)) \\
                    % Start of broadcast
                    %\tab\tab (\pactproofs_{\asset^{\star}}, \pactproofs_{\consistency'}, \pactproofs_{\equivalence}) \gets \broadcastfunction(p, \commitment, \token, \pacttxs)\\
                    % End of broadcast
                    \tab\tab\pactproofs_{p,a} := (\pactproofs^{\consistency}) \\
                    \tab\tab  \pactproofs_a.append(\pactproofs_{p,a}) \\
                    \tab\tab\pacttx_{p,a} := (\commitment, \token) \\ 
                    \tab\tab \pacttx_a.append(\pacttx_{p,a})\\
                    \tab \pacttx.append(\pacttx_a), \pactproofs.append(\pactproofs_a) \\
                    \pacttx := (H(\pacttx),\pacttx) \\ 
                    \dataend \gets \broadcastfunction(\pacttx)\\
                    \textnormal{for each asset $a$ and each participant $p$:}\\
                    \tab \textnormal{parse } (\commitment'_{p,a}, \token'_{p,a},\pactproofs^{\asset^{\star}}_{p,a}, \pactproofs^{\consistency'}_{p,a}, \pactproofs^{\equivalence}_{p,a}) \textnormal{ from } \dataend \\ 
                    \tab \textnormal{update } \pactproofs_{p,a} := (\pactproofs^{\asset^{\star}}_{p,a}, \pactproofs^{\consistency}_{p,a}, \pactproofs^{\consistency'}_{p,a}, \pactproofs^{\equivalence}_{p,a}) \textnormal{ in } \pactproofs \\
                    \tab \textnormal{update } \pacttx_{p,a} := (\commitment_{p,a}, \commitment'_{p,a}, \token_{p,a}, \token'_{p,a}) \textnormal{ in } \pacttx\\
                    \ret (\pacttx, \pactproofs)\\
                }
             &
            \makecell[l]{
           \underline{\boldsymbol{\pactverify\leftbrc \pacttx, \pactproofs,\pacttxs,\pactpks \rightbrc }} \\
            (g,h) \gets \publicpar \\
            b \gets 1\\
            \textnormal{for each asset $a$ in $\pacttx$:}\\
            \tab b_{\balance} := (\Pi_{p=1}^{|\pactpks|}\commitment_{p,a} \stackrel{?}{=} 1) \\
            \tab \textnormal{for each party $p$ in $\pacttx_a$:}\\
                \tab\tab (\pactproofs^{\asset^{\star}},\pactproofs^{\consistency}, \pactproofs^{\consistency'},\pactproofs^{\equivalence}) := \pactproofs_{p,a}\\
                \tab\tab(\commitment, \token, \commitment', \token') := \pacttx_{p,a} \\
                \tab\tab b_{\consistency} := \nizkverify_{\lang_{\relation_{\consistency}}}(\crs, (\commitment,\token), \pactproofs^{\consistency}) \\
                \tab\tab b_{\consistency'} := \nizkverify_{\lang_{\relation_{\consistency}}}(\crs, (\commitment',\token'), \pactproofs^{\consistency'}) \\
                \tab\tab b_{\equivalence} := \nizkverify_{\lang_{\relation_{\equivalence}}}(\crs, (\commitment/\commitment',\token/\token'), \pactproofs^{\equivalence}) \\
                 \tab\tab \commitment_{\Pi} := (\Pi_{i=1}^{|\pacttxs|} \commitment_{i,p}) \cdot \commitment \\
                 \tab\tab \commitment'_{\Pi}, \pactproofs^{\asset} := \pactproofs^{\asset^{\star}}\\
                \tab\tab b_{\asset} := \nizkverify_{\lang_{\relation_{\asset}}}(\crs, (\commitment_{\Pi}, \commitment'_{\Pi} ,2^N), \pactproofs^{\asset})\\
                \tab\tab b_{p,a} := b_{\balance} \wedge  b_{\consistency} \wedge  b_{\consistency'} \wedge  b_{\equivalence} \wedge  b_{\asset} \\
                \tab\tab b := b \wedge b_{p,a} \\
            \ret b
            \\
            } \\
            % \hline
            %    \makecell[l]{
            % \underline{\boldsymbol{\broadcastfunction(p, \commitment, \token, \pactaux)}} \\
            %     (\pactproofs_{\asset^{\star}}, \pactproofs_{\consistency'}, \pactproofs_{\equivalence}) \gets \fillasset_{p}(\commitment,\token, \pactaux) \\
            %     % End Proof of asset
            %     \ret (\pactproofs_{\asset^{\star}}, \pactproofs_{\consistency'}, \pactproofs_{\equivalence}) \\
            %     % End of broadcast
            % \\
           \hline
           \makecell[l]{
            \underline{\boldsymbol{\broadcastfunction(\pacttx)}} \\
                (\pactaux, \pacttx) := \pacttx \\
                \dataend = []\\
                \textnormal{for each participant $p$ and asset $a$:} \\
                \tab \textnormal{parse } \commitment_{p,a}, \token_{p,a} \textnormal{ from } \pacttx \\
                \tab \dataend_{p,a} \gets \fillasset_{p}(\commitment_{p,a},\token_{p,a}, \pactaux) \\
                \tab\dataend.append(\dataend_{p,a}) \\
                %\pactproofs^{\asset^{\star}}_{p}, \pactproofs^{\consistency'}_{p}, \pactproofs^{\equivalence}_{p}
                % End Proof of asset
                \ret \dataend\\
                % End of broadcast
                \\
            \hline \\
            \underline{\boldsymbol{\fillasset(\commitment, \token, \pactaux)}} \\
            (\pactsk_p,\pactpk_p, \pactpolicy_p) := \stat_p \\
            \pactcoinvalue_p := \extract (\commitment, \token, \pactsk_p) \\
            \ret \pactconsent\leftbrc \pactcoinvalue_p, \pactpolicy_p, \pactsk_p, (\commitment, \pactpk_{p}, \pactaux) \rightbrc \\
            }
            &
            \makecell[l]{
            \\
            \underline{\boldsymbol{\pactconsent\leftbrc \pactcoinvalue_p, \pactpolicy_p, \pactsk_p, \pactaux \rightbrc}} \\
             % Start of broadcast
             \commitment',\coinrand' := \pactmint(\pactcoinvalue_{p}) \\
            \token':= \pactpk_{p}^{r'} \\
             \pactproofs^{\consistency'} := \nizkprove_{\lang_{\relation_{\consistency}}}(\crs, (\commitment',\token'), (\pactcoinvalue_p,\coinrand')) \\
             \pactproofs^{\equivalence} := \nizkprove_{\lang_{\relation_{\equivalence}}}(\crs, (\commitment/\commitment',\token/\token'), (\pactsk_p)) \\
            % Proof of asset
            \commitment_{\Pi} := (\Pi_{i=1}^{|\pacttxs|} \commitment_{i,p}) \cdot \commitment \\
            \pactcoinvalue_{\sum} :=  (\sum_{i=1}^{|\pacttxs|} \pactcoinvalue_{i,p}) + \pactcoinvalue_p \\
            \commitment'_{\Pi}, \coinrand_{\Pi} := \pactmint(\pactcoinvalue_{\sum}) \\
            \pactproofs^{\asset} := \nizkprove_{\lang_{\relation_{\asset}}}(\crs, (\commitment_{\Pi}, \commitment'_{\Pi}, 2^N), (\pactcoinvalue_{\sum}, \coinrand_{\Pi}, \pactsk_p))\\
            \pactproofs^{\asset^{\star}} := (\commitment'_{\Pi},\pactproofs^{\asset}) \\
            % End Proof of asset
            \ret (\commitment',\token',\pactproofs^{\asset^{\star}}, \pactproofs^{\consistency'}, \pactproofs^{\equivalence}) \textnormal{ if }\pactpolicy(\pactcoinvalue_p,\pactaux)=1 \\
            \textnormal{else }\ret \ground 
                % End of broadcast
            \\
            }
            \\
        \end{tabularx}
    \end{subfigure}
    \vspace{0.25cm}
    \caption{The pseudocode of PADL for the algorithms in Definition \ref{definition:pact}. Without loss of generality, we assumed transaction value array in $\pactvalue$ contains entries for all $\pactpks$, otherwise the array is zero-padded in appropriate missing entry. Algorithm $\broadcastfunction$ is executed by broadcaster, $\fillasset_p$ and $\pactconsent_p$ are executed by $p$-participant. }
    \label{table:padl_algo}
    \scriptsize
\end{table}

\subsection{Privacy Preserving Auditing of Confidential Assets}
\label{auditing_section}
Previous works, such as zkLedger~\cite{narula2018zkLedger}, demonstrated the use of the homomorphic property of Pedersen commitments and audit tokens for auditing known assets.
In the context of confidential assets, we can generate privacy-preserving audits that can be used by external auditors, or being used directly as conditions for transactions without sharing the sum of values. For example, it can be used for liquidity and credit risk, or interest rate return payments. 
We show the correctness for proofs of basic auditing of an asset balance, asset liquidity and inter-transactions rate proof without revealing transaction history, thereby  achieving PADL privacy. The auditing ZKP can simply extend the transaction cell structure or being queried separately on the ledger.
%\wz{will fix symbol below.}
\subsubsection{Basic Asset Balance.}
\label{audit_balance}

Auditing any asset balance is done by proving the equivalence between the product of commitments signed by a secret key and the product of tokens as follows. 
Participant $p$ ($prover$), shares with an auditor ($verifier$) the claimed balance $\pactcoinvalue_{p,a}$ for asset $a$.
The $verifier$ calculates product of commitments and product of tokens (for a particular asset and participant),
$\prod_{t\in Txs}\commitment_{t,p,a}/g^{\pactcoinvalue_{p,a}}:= c_1$ and $\prod_{t\in Txs}\token_{t,p,a} := c_2$, respectively.

The $prover$ sends the zero knowledge proof \cite{AFRICACRYPT:Maurer09} for the equivalence, $dlog_{c_1}c_2 = dlog_{h}{\pactpk_p}$, where both right and left terms open to $\pactsk$. The $verifier$ verifies the proof, and also verifies $c_1$ and $c_2$ with the ledger. 
It is also easily inferred that using the basic auditing of the asset sum of transactions values, ratio, variance, and concentrations can be calculated as mentioned previously~\cite{narula2018zkLedger} for non-confidential assets. Next we show proofs that do not share the sum of values.

\subsubsection{Asset Liquidity}
\label{audit_liquidity}

PADL can maintain a large amount of confidential assets. When used in a debit market scenario, these assets can be bonds or other type of loan asset tokens as well as currency tokens. The PADL transaction scheme allows a participant to be audited for credit or liquidity of an asset by generating a `liquidity' proof. This can be done by calculating fractional ratios among different assets for multiple auditors and multiple participants without revealing any further information about the balance of the assets.
The participant proves that given a Ledger, and using NIZK proof that the ratio of its investment in an asset to its total assets value is lower than a desired threshold, as follows:

For a given rational number, $f \in \mathbb{Q}$ and two integers, $D$, $N$ such that $D/N=f$, participant, $p$ proves for index asset, $a^*$ its liquidity that:
\begin{equation*}
    \frac{\sum_{t \in Txs} \pactcoinvalue_{t,p,a^*}}{\sum_{a\in A}\sum_{t \in Txs}\pactcoinvalue_{t,p,a}}<f
\end{equation*}
The $prover$ and $verifier$ calculate from the ledger the two points, 

$\prod_{t \in Txs} \commitment'_{t,p,a^*} := c_1$,
$\prod_{a\in A}\prod_{t \in Txs} \commitment'_{t,p,a} := c_2$, and,
$c_r = c_2^D/c_1^N$.

The $prover$ calculates:
   $ \pactcoinvalue_r=D\sum_{a\in A}\sum_{t \in Txs}\pactcoinvalue_{t,p,a}-N\sum_{t \in Txs} \pactcoinvalue_{t,p,a^*}$.

The $prover$ sends a range-proof, $\pactproofs^{AL}$ with the ratio commit, $c_r$, to prove that $\pactcoinvalue_r$ lies in interval $[0...2^n-1]$.
\textit{verifier} verifies the range-proof, $\pactproofs^{AL}$ and that $c_2^D/c_1^N= c_r$ using the ledger and $(a^*,D,N)$.

\subsubsection{Inter-transactions Rate Proof}
\label{audit_intertransaction}

In decentralized finance applications, transactions for multiple investors may consider fractional interest rates. These rates can be proved as part of the transaction verification process or for auditing purposes using $\Sigma$-protocol proofs between cells in the ledger. This approach can be used to conceal investors' strategies while still being paid coupon payment in a verified manner (as discuss on Sec. \ref{ch:usecases}, or to be used for concentration tractability in auditing, or trading in secondary markets.
For a given rational fraction, $Rate$, participant $p$ proves for asset $a$ that:
\begin{equation*}
    \frac{\sum_{t \in txs_1 \subset T} \pactcoinvalue_{t,p,a}=\Sigma v_1}{\sum_{t \in txs_2 \subset T}\pactcoinvalue_{t,p,a}=\Sigma v_2}=Rate,
\end{equation*}
where $\sum_{t \in txs_1 \subset T} \pactcoinvalue_{t,p,a}=\Sigma v_1$ and $\sum_{t \in txs_2 \subset T}\pactcoinvalue_{t,p,a}=\Sigma v_2$ are the sum of values over two subsets of transactions.
Here the $verifier$ also provides two integers, $D$, $N$ such that $D/N=Rate$.
The $prover$ and $verifer$ calculate from the ledger:
\begin{align*}
\prod_{t\in txs_1} \commitment_{t,p,a} := c_1, \prod_{t\in txs_2} \token_{t,p,a} := \tau_1, \quad  
\prod_{t\in txs_2} \commitment_{t,p,a} := c_2, \prod_{t\in txs_2} \token_{t,p,a} := \tau_2
\end{align*}
and $c = c_1^N \cdot c_2^{-D}$,  and $\tau = \tau_1^N \cdot \tau_2^{-D}$,
and then $prover$ provides the proof of knowledge of a solution ($\pactsk$) that solves the equivalence DLP as follows:

% v1/v2=D/N = f, c1/c2
For correctness, $c,\tau$ are shown to be:
\begin{align*}
c=g^{N\Sigma v_1-D\Sigma v_2}h^{N\Sigma r_1-D\Sigma r_2}, \quad \tau=\pactpk^{N\Sigma r_1-D\Sigma r_2}
\end{align*}
where $\Sigma r_1=\sum_{t \in txs_1 \subset T} \coinrand_{t,p,a}$ and $\Sigma r_2=\sum_{t \in txs_2 \subset T} \coinrand_{t,p,a}$. 
If $\sum v_1/ \sum v_2 = Rate$, then $ c = h^{N\Sigma r_1-D\Sigma r_2}$
and $prover$ provides a proof for the equivalence of the DLP in $F_P$, $dlog_{c} \tau \equiv dlog_{c} c^{\pactsk}$.
Otherwise, $c^* = g^{v^*}h^{N\Sigma r_1-D\Sigma r_2}$ and finding a solution for $dlog_{c^*} \tau$ is intractable. 

\subsection{Full Auditing and Customized Traceability}
Our system provides a straightforward way to map auditing as required by regulators or parties. This is done by appending to a cell in a transaction, an additional $\token$, but with the auditing party's public key. The 'settlement' bank use-case (Sec. \ref{ch:usecases}) shows an example of such auditing, but this flexibility becomes even more powerful in a complex map of trust between parties with multiple assets. For example, in the debit market where loans contain multiple assets besides cash assets, and requires trusted parties as banks, custodians, and brokers.

\section{Security and Privacy Evaluation}
\label{sec:eval}
\label{sec:security}
In this section, the proposed multi-asset transaction scheme PADL is analyzed in term of security, privacy and performance metrics.  
% \subsection{Security and Privacy Evaluation}
PADL should satisfy integrity preserving and anonymity preserving properties.
A transaction scheme for a ledger-based system is considered integrity preserving if it maintains the state of the account balance properly. In $\pact$, it is required that the transaction scheme should never allows spending more than the account balance (double spending in the PADL context means spending more than the account balance), never allows asset creation in a transaction and never allows spending of balance from other accounts without owner endorsement. The aforementioned notions are captured in Definition \ref{def:invariant}. 
Integrity property is defined to preserve the invariant defined in Definition \ref{def:invariant} and is formally defined in the Appendix. 
% To facilitate the usecase of permissioned distributed ledger participated by financial institutions, 
In contrast to Ring-CT in \cite{ESORICS:SALY17} where the recipient address is known, $\pact$ scheme requires that the recipient and sender are anonymized because recipient address is fixed in $\pact$ unlike stealth-address used in Monero transaction. In the context of PADL, the participant information, asset information and transaction value are anonymized among a anonymity set.
The formal definitions for integrity and anonymity and their proofs can be found in Appendix \ref{appendix:security}. Informally, integrity property preserves invariant. Asset overspending is prevented by the soundness of proof of asset and binding property of commitment scheme. Endorsement forgery is prevented by the discrete log hardness and zero-knowledge and soundness of proof. Transaction imbalance is prevented due to the hardness of computing discrete log of the commitment key $h$. Therefore, the system ensured that no asset is created and endorsements from account owners are required for the transaction to go through. In addition, PADL satisfies anonymity because the commitment hides the transaction value and the proofs generated are zero-knowledge. In addition, the tokens are just randomness. Therefore, the transaction scheme does not leak any information about the sender and recipient.
In essence, PADL which is an instantiation of $\pact$ scheme has the following security and privacy guarantee:

\begin{definition}[Invariant] For the $\pact$ transaction scheme to be integrity preserving, it is required that the following invariant hold:
\label{def:invariant}
    \begin{itemize}
        \item \textbf{I1 Asset balance}: For each party $p$ and each asset $a$ in a transaction $\pacttx$, the spending (negative) amount $\pactcoinvalue$ must be smaller than the current asset value $\abs{\pactcoinvalue}\leq \sum \pactcoinvalue_{i,p, a}$ where $i$ ranges over all transactions before the current transaction.
        \item \textbf{I2 Account Owner Endorsement}: For each account $\pactpk_p$ in a transaction $\pacttx$, the endorser should possess the corresponding account secret $\pactsk_p$.
        \item \textbf{I3 Transaction balance}: For a transaction with some transaction row index $t$ and each asset $a$, it must be the case that the transaction value summed to zero, $\sum_{p=1}^{|\pactpks|}{\pactcoinvalue_{t,p,a}} = 0$.
    \end{itemize}
    \label{definition:invariant}
\end{definition}

\begin{theorem}
Let the underlying proofs system used be zero-knowledge and sound, the commitment scheme used has binding property and is homomorphic, and discrete log is hard, then PADL presented in Table \ref{table:padl_algo} satisfied integrity defined in Definition \ref{definition:exp_balance} at Appendix \ref{appendix:security}.
\label{theo:balance}
\end{theorem}
%Integrity property preserve defined invariant. Asset overspending is prevented by the soundness of proof of asset and binding property of commitment scheme. Endorsement forgery is prevented by the discrete log hardness and zero-knowledge and soundness of proof. Transaction imbalance is prevented due to the hardness of computing discrete log of the commitment key $h$. Therefore, the system ensured that no asset is created and endorsements from account owners are required for the transaction to go through.

\begin{theorem}
Let the commitment scheme used has hiding property and underlying proof system used is zero-knowledge, then PADL presented in Table \ref{table:padl_algo} satisfied anonymity defined in Definition \ref{definition:exp_anon} at Appendix \ref{appendix:security}.
\label{theo:anonymity}
\end{theorem}
\noindent\textit{Proofs for both theorem above are provided in Appendix \ref{appendix:security}.} %The commitment hides the transaction value and the proofs generated are zero-knowledge. In addition, the tokens are just randomness. Therefore, the transaction scheme does not leak any information about the sender and recipient.

\section{Performance Analysis}
\label{implementation}
We present PADL as a comprehensive package for constructing ledgers (link to code\textsuperscript{\textdagger}), both as standalone systems and integrated with smart contracts. The implementation of PADL is designed to be modular, enabling the development of new financial use cases. It employs RUST for cryptographic protocols and primitives to achieve high performance and tested code, but the ledger construction is interfaced and packaged with Python to obtain flexibility and accessibility. 
\vspace{-0.5cm}
\begin{table}[h] 
\begin{center}
\begin{tabular}{llccc}
\hline
 Ledger & Type (section) & Size (bytes) & Time/txn (sec) & Assets/Banks \\ [0.5ex] 
 \hline
  Simple Exchange & Tx (Sec. \ref{simple_exchange_ledger}) & $4,704$ & 0.34 $\pm 0.01$ & 2/2 \\ 
 
  Settlement Trusted Bank & Tx (Sec. \ref{settlement_ledger_trusted})   & $3,726$ & 0.21 $\pm 0.001$ & 1/3 \\
  Bond Market & Tx (Sec. \ref{bond_market_ledger}) & $16,464$ & 0.41 $\pm 0.03$ & 2/7\\
  proofs+commits & Cell (Sec. \ref{methods}) & $1,176$ & - & 1/1 \\
   asset balance & proofs (Sec. \ref{audit_balance}) & $98$ & - & -\\
 inter-transactions & proofs (Sec. \ref{audit_intertransaction}) & $98$ & - & -\\
 asset liquidity & proofs (Sec. \ref{audit_liquidity}) & $688$ & - & -\\
\end{tabular}
\end{center}
\caption{A table with the main experimental use-cases from Sec. \ref{ch:usecases} and individual auditing proofs. It includes the size of the cell, and the size and time of the transaction, and the number of assets/banks. We performed 10 runs for each use-case to get the average Time/txn and the $\pm$ confidence intervals.}
\vspace{-1cm}
\label{table:sizecell}
\end{table}

A performance analysis was conducted for the PADL implementation, alongside the zkLedger package and its default configuration. We summarize the transaction size and speed in PADL, the size of auditing proofs, and provide comparative metrics with zkLedger where feasible. From the results in \cite{narula2018zkLedger}, it is observed that the size of the range proof is $3.9$KB, compared to $1.2$KB for the entire commits and proofs in PADL (range proof for $2^{32}-1$). Besides the utilization of Bulletproofs in PADL, another reason for the reduced size in PADL is related to the Disjunctive proof present in zkLedger.

Table \ref{table:sizecell} also illustrates the execution times for the examples discussed in Sec. \ref{ch:usecases}. For the bond market example involving 2 assets and 7 participants, PADL requires approximately $0.41$ seconds on average per transaction. Auditing proofs primarily involve homomorphic arithmetic on the set of transactions in the ledger and the construction of a $\Sigma$-protocol transcript proof. The results demonstrate the practicality and scalability of the solution within financial systems. Figure \ref{fig:benchm} depicts the performance characteristics of both systems under test. As observed in Fig. \ref{fig:benchm}a, there is a linear increase in execution time (seconds, y-axis) relative to the number of assets (x-axis) for zkLedger when compared to PADL. Due to the asynchronous generation of proof of Asset, PADL exhibits better scalability with an increasing number of banks or assets, whereas zkLedger incurs a computationally costly overhead with significant increments in either parameter.

\begin{figure}[!htbp]
    \centering
    \vspace{-0.5cm}
    \includegraphics[width=0.49\textwidth]{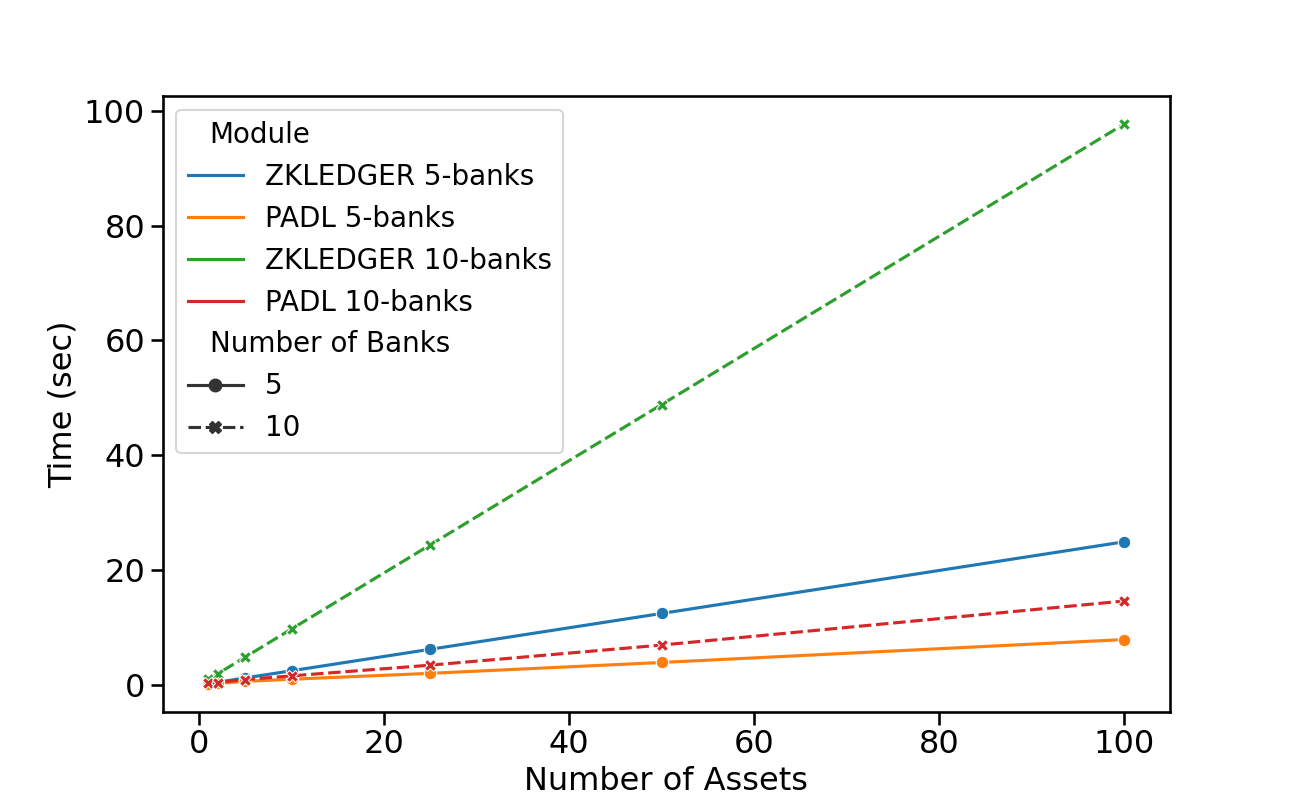} % first figure itself
    \includegraphics[width=0.49\textwidth]{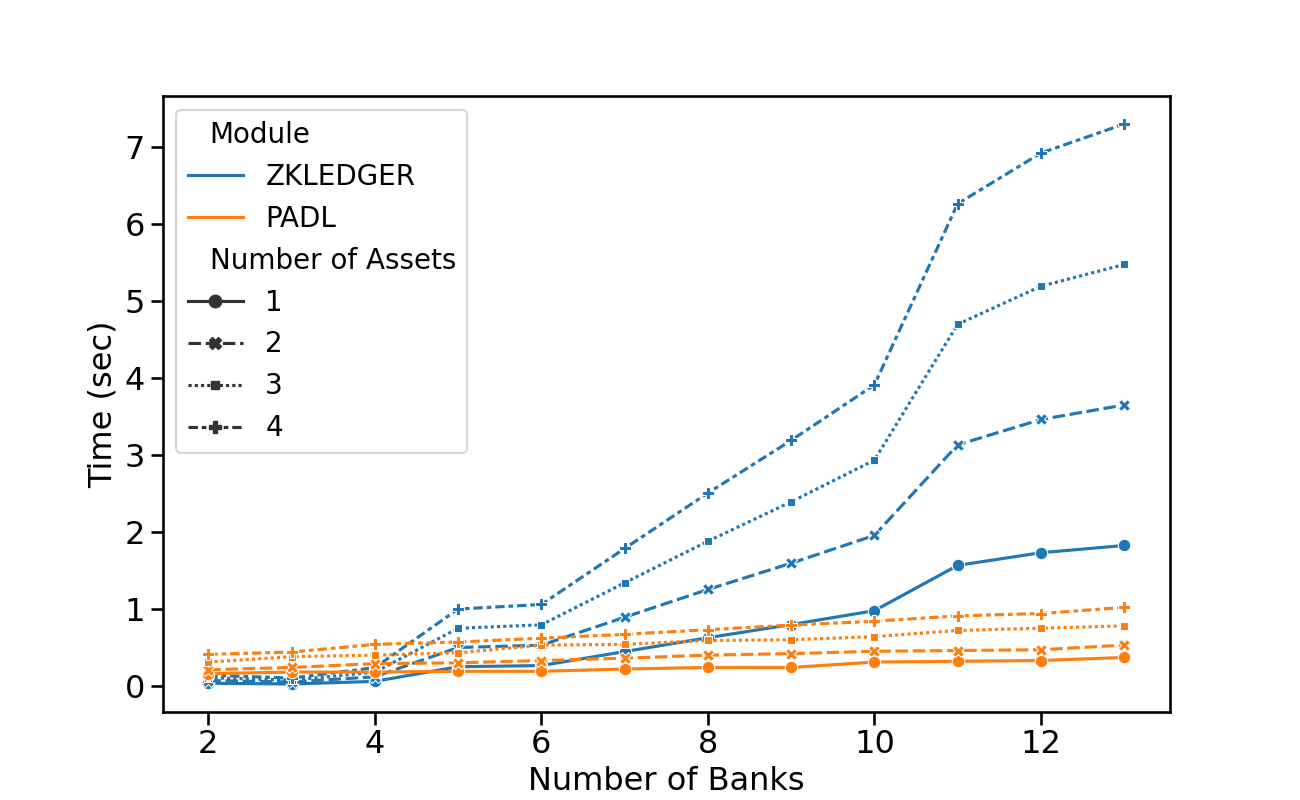} % second figure itself
    \caption{PADL speed benchmark vs zkLedger algorithm. The y-axis in both figures shows the time (s) it takes for PADL and zkLedger to complete 4 transactions (proving and verification). In subplot \textbf{a} we show over different number of assets in the x-axis and in subplot \textbf{b} over different number of participants-banks. The type of lines and markers show different configurations for the benchmark. The average error (standard deviation in seconds) }
    \label{fig:benchm}
    \vspace{-0.5cm}
\end{figure}

\section{Conclusions and Future Work}
PADL proposes a combined system for a transaction scheme to hide values and balances of assets in distributed ledgers. The ledger focuses on developing financial institutes applications with confidential multi-assets. The proposed scheme with the security evaluation of integrity and anonymity of transaction, shows the general purpose usage of the transaction scheme with privacy preserving auditing and/or `full' auditing. Performance evaluation of the ledger shows that the PADL transaction scheme is scaled well with number of participants due to asynchronous proof generation and batch verification, as well as shorter transaction size when comparing to previous `table' ledgers, such as zkLedger. This is also due to a simplified cryptography scheme used here with complementary commits and equivalence proof. Future direction to expand the framework is to include further set of auditing proofs related to membership and digital identity, which is required for; categories,sanction list membership, as well as additional proofs in the flow of trades of various complex financial markets.

\section{Acknowledgment}
We would like to acknowledge Onyx and Blockchain Engineer team at JPMorgan Chase for their support throughout this project. In particular, we thank Sudhir Upadhyay for his support and advice through-out the work. We also appreciate the assistance of Sai Valiveti, Joe Leung, and Imran Bashir, as well as the support provided by Suresh Shetty. Additionally, we thank Lianna Zhao for her help in the revision process.

\section{Disclaimer}
This paper was prepared for informational purposes by the Global Technology Applied Research Center of JPMorgan Chase \& Co. This paper is not a product of the Research Department of JPMorgan Chase \& Co. or its affiliates. Neither JPMorgan Chase \& Co. nor any of its affiliates makes any explicit or implied representation or warranty and none of them accept any liability in connection with this paper, including, without limitation, with respect to the completeness, accuracy, or reliability of the information contained herein and the potential legal, compliance, tax, or accounting effects thereof. This document is not intended as investment research or investment advice, or as a recommendation, offer, or solicitation for the purchase or sale of any security, financial instrument, financial product or service, or to be used in any way for evaluating the merits of participating in any transaction.

%%
%% Bibliography
%%

%% Please use bibtex, 

\bibliographystyle{abbrv}% the mandatory bibstyle
\bibliography{references, Arxiv/abbrev1, Arxiv/crypto, extra_refences_PADL}

\begin{thebibliography}{10}

\bibitem{alarab2020competence}
I.~Alarab, S.~Prakoonwit, and M.~I. Nacer.
\newblock Competence of graph convolutional networks for anti-money laundering
  in bitcoin blockchain.
\newblock In {\em Proceedings of the 2020 5th international conference on
  machine learning technologies}, pages 23--27, 2020.

\bibitem{ames2017ligero}
S.~Ames, C.~Hazay, Y.~Ishai, and M.~Venkitasubramaniam.
\newblock Ligero: Lightweight sublinear arguments without a trusted setup.
\newblock In {\em Proceedings of the 2017 acm sigsac conference on computer and
  communications security}, pages 2087--2104, 2017.

\bibitem{androulaki2018hyperledger}
E.~Androulaki, A.~Barger, V.~Bortnikov, C.~Cachin, K.~Christidis, A.~De~Caro,
  D.~Enyeart, C.~Ferris, G.~Laventman, Y.~Manevich, et~al.
\newblock Hyperledger fabric: a distributed operating system for permissioned
  blockchains.
\newblock In {\em Proceedings of the thirteenth EuroSys conference}, pages
  1--15, 2018.

\bibitem{Androulaki2019}
E.~Androulaki, J.~Camenisch, A.~D. Caro, M.~Dubovitskaya, K.~Elkhiyaoui, and
  B.~Tackmann.
\newblock {Privacy-preserving auditable token payments in a permissioned
  blockchain system}.
\newblock {\em Cryptology ePrint Archive}, 2019.

\bibitem{sasson2018zkstarks}
E.~Ben{-}Sasson, I.~Bentov, Y.~Horesh, and M.~Riabzev.
\newblock Scalable, transparent, and post-quantum secure computational
  integrity.
\newblock {\em {IACR} Cryptol. ePrint Arch.}, page~46, 2018.

\bibitem{sasson2014zksnarks}
E.~Ben-Sasson, A.~Chiesa, E.~Tromer, and M.~Virza.
\newblock Succinct {Non-Interactive} zero knowledge for a von neumann
  architecture.
\newblock In {\em 23rd USENIX Security Symposium (USENIX Security 14)}, pages
  781--796, San Diego, CA, Aug. 2014. USENIX Association.

\bibitem{bhadauria2020ligero++}
R.~Bhadauria, Z.~Fang, C.~Hazay, M.~Venkitasubramaniam, T.~Xie, and Y.~Zhang.
\newblock Ligero++: A new optimized sublinear iop.
\newblock In {\em Proceedings of the 2020 ACM SIGSAC Conference on Computer and
  Communications Security}, pages 2025--2038, 2020.

\bibitem{blum2019non}
M.~Blum, P.~Feldman, and S.~Micali.
\newblock Non-interactive zero-knowledge and its applications.
\newblock In {\em Providing Sound Foundations for Cryptography: On the Work of
  Shafi Goldwasser and Silvio Micali}, pages 329--349. Association for
  Computing Machinery, New York, NY, USA, 2019.

\bibitem{SP:BBBPWM18}
B.~B{\"u}nz, J.~Bootle, D.~Boneh, A.~Poelstra, P.~Wuille, and G.~Maxwell.
\newblock Bulletproofs: Short proofs for confidential transactions and more.
\newblock In {\em 2018 {IEEE} Symposium on Security and Privacy}, pages
  315--334. {IEEE} Computer Society Press, May 2018.

\bibitem{cecchetti2017solidus}
E.~Cecchetti, F.~Zhang, Y.~Ji, A.~Kosba, A.~Juels, and E.~Shi.
\newblock Solidus: Confidential distributed ledger transactions via pvorm.
\newblock In {\em Proceedings of the 2017 ACM SIGSAC Conference on Computer and
  Communications Security}, pages 701--717, 2017.

\bibitem{chatzigiannis2021miniledger}
P.~Chatzigiannis and F.~Baldimtsi.
\newblock Miniledger: compact-sized anonymous and auditable distributed
  payments.
\newblock In {\em European Symposium on Research in Computer Security}, pages
  407--429. Springer, 2021.

\bibitem{cramer1996modular}
R.~Cramer.
\newblock Modular design of secure yet practical cryptographic protocols.
\newblock {\em Ph. D.-thesis, CWI and U. of Amsterdam}, 2, 1996.

\bibitem{CCS:DBBCB15}
G.~G. Dagher, B.~B{\"u}nz, J.~Bonneau, J.~Clark, and D.~Boneh.
\newblock Provisions: Privacy-preserving proofs of solvency for bitcoin
  exchanges.
\newblock In I.~Ray, N.~Li, and C.~Kruegel, editors, {\em ACM CCS 2015: 22nd
  Conference on Computer and Communications Security}, pages 720--731. {ACM}
  Press, Oct. 2015.

\bibitem{Ding2020a}
D.~Ding, K.~Li, L.~Jia, Z.~Li, J.~Li, and Y.~Sun.
\newblock {Privacy protection for blockchains with account and multi-asset
  model}.
\newblock {\em China Communications}, 16(6):69--79, jul 2020.

\bibitem{C:FiaSha86}
A.~Fiat and A.~Shamir.
\newblock How to prove yourself: {Practical} solutions to identification and
  signature problems.
\newblock In A.~M. Odlyzko, editor, {\em Advances in Cryptology --
  {CRYPTO}'86}, volume 263 of {\em Lecture Notes in Computer Science}, pages
  186--194. Springer, Berlin, Heidelberg, Aug. 1987.

\bibitem{gao2019private}
Z.~Gao, L.~Xu, K.~Kasichainula, L.~Chen, B.~Carbunar, and W.~Shi.
\newblock Private and atomic exchange of assets over zero knowledge based
  payment ledger.
\newblock {\em arXiv preprint arXiv:1909.06535}, 2019.

\bibitem{Garman2016}
C.~Garman, M.~Green, and I.~Miers.
\newblock {Accountable Privacy for Decentralized Anonymous Payments}.
\newblock {\em Cryptology ePrint Archive}, 2016.

\bibitem{goldwasser2019knowledge}
S.~Goldwasser, S.~Micali, and C.~Rackoff.
\newblock The knowledge complexity of interactive proof-systems.
\newblock {\em SIAM Journal of Computing}, pages 186--208, 1989.

\bibitem{groth2010short}
J.~Groth.
\newblock Short non-interactive zero-knowledge proofs.
\newblock In {\em Advances in Cryptology-ASIACRYPT 2010: 16th International
  Conference on the Theory and Application of Cryptology and Information
  Security, Singapore, December 5-9, 2010. Proceedings 16}, pages 341--358.
  Springer, 2010.

\bibitem{EC:GroKoh15}
J.~Groth and M.~Kohlweiss.
\newblock One-out-of-many proofs: Or how to leak a secret and spend a coin.
\newblock In E.~Oswald and M.~Fischlin, editors, {\em Advances in Cryptology --
  {EUROCRYPT}~2015, Part~II}, volume 9057 of {\em Lecture Notes in Computer
  Science}, pages 253--280. Springer, Berlin, Heidelberg, Apr. 2015.

\bibitem{EC:GroOstSah06}
J.~Groth, R.~Ostrovsky, and A.~Sahai.
\newblock Perfect non-interactive zero knowledge for {NP}.
\newblock In S.~Vaudenay, editor, {\em Advances in Cryptology --
  {EUROCRYPT}~2006}, volume 4004 of {\em Lecture Notes in Computer Science},
  pages 339--358. Springer, Berlin, Heidelberg, May~/~June 2006.

\bibitem{jeong2023azeroth}
G.~Jeong, N.~Lee, J.~Kim, and H.~Oh.
\newblock Azeroth: Auditable zero-knowledge transactions in smart contracts.
\newblock {\em IEEE Access}, 2023.

\bibitem{Jia2024}
W.~Jia, T.~Xie, and B.~Wang.
\newblock {A privacy-preserving scheme with multi-level regulation compliance
  for blockchain}.
\newblock {\em Scientific Reports 2024 14:1}, 14(1):1--14, jan 2024.

\bibitem{kang2019fabzk}
H.~Kang, T.~Dai, N.~Jean-Louis, S.~Tao, and X.~Gu.
\newblock Fabzk: Supporting privacy-preserving, auditable smart contracts in
  hyperledger fabric.
\newblock In {\em 2019 49th Annual IEEE/IFIP International Conference on
  Dependable Systems and Networks (DSN)}, pages 543--555. IEEE, 2019.

\bibitem{khattak2020dynamic}
H.~A. Khattak, K.~Tehreem, A.~Almogren, Z.~Ameer, I.~U. Din, and M.~Adnan.
\newblock Dynamic pricing in industrial internet of things: Blockchain
  application for energy management in smart cities.
\newblock {\em Journal of Information Security and Applications}, 55:102615,
  2020.

\bibitem{aggelos2022peredi}
A.~Kiayias, M.~Kohlweiss, and A.~Sarencheh.
\newblock Peredi: Privacy-enhanced, regulated and distributed central bank
  digital currencies.
\newblock Cryptology ePrint Archive, Paper 2022/974, 2022.
\newblock \url{https://eprint.iacr.org/2022/974}.

\bibitem{li2019toward}
Y.~Li, W.~Susilo, G.~Yang, Y.~Yu, X.~Du, D.~Liu, and N.~Guizani.
\newblock Toward privacy and regulation in blockchain-based cryptocurrencies.
\newblock {\em IEEE Network}, 33(5):111--117, 2019.

\bibitem{li2019traceable}
Y.~Li, G.~Yang, W.~Susilo, Y.~Yu, M.~H. Au, and D.~Liu.
\newblock Traceable monero: Anonymous cryptocurrency with enhanced
  accountability.
\newblock {\em IEEE Transactions on Dependable and Secure Computing},
  18(2):679--691, 2019.

\bibitem{AFRICACRYPT:Maurer09}
U.~M. Maurer.
\newblock Unifying zero-knowledge proofs of knowledge.
\newblock In B.~Preneel, editor, {\em AFRICACRYPT 09: 2nd International
  Conference on Cryptology in Africa}, volume 5580 of {\em Lecture Notes in
  Computer Science}, pages 272--286. Springer, Berlin, Heidelberg, June 2009.

\bibitem{meralli2020}
S.~Meralli.
\newblock Privacy-preserving analytics for the securitization market: a
  zero-knowledge distributed ledger technology application.
\newblock {\em Financial Innovation}, 6(1):7, 2020.

\bibitem{nakamoto2008}
S.~Nakamoto.
\newblock Bitcoin: A peer-to-peer electronic cash system.
\newblock {\em SSRN Electronic Journal}, 2008.

\bibitem{narula2018zkLedger}
N.~Narula, W.~Vasquez, and M.~Virza.
\newblock $\{$zkLedger$\}$:$\{$Privacy-Preserving$\}$ auditing for distributed
  ledgers.
\newblock In {\em 15th USENIX Symposium on Networked Systems Design and
  Implementation (NSDI 18)}, pages 65--80, 2018.

\bibitem{noether2016ring}
S.~Noether, A.~Mackenzie, et~al.
\newblock Ring confidential transactions.
\newblock {\em Ledger}, 1:1--18, 2016.

\bibitem{C:Pedersen91}
T.~P. Pedersen.
\newblock Non-interactive and information-theoretic secure verifiable secret
  sharing.
\newblock In J.~Feigenbaum, editor, {\em Advances in Cryptology --
  {CRYPTO}'91}, volume 576 of {\em Lecture Notes in Computer Science}, pages
  129--140. Springer, Berlin, Heidelberg, Aug. 1992.

\bibitem{sasson2014zerocash}
E.~B. Sasson, A.~Chiesa, C.~Garman, M.~Green, I.~Miers, E.~Tromer, and
  M.~Virza.
\newblock Zerocash: Decentralized anonymous payments from bitcoin.
\newblock In {\em 2014 IEEE symposium on security and privacy}, pages 459--474.
  IEEE, 2014.

\bibitem{C:Schnorr89}
C.-P. Schnorr.
\newblock Efficient identification and signatures for smart cards.
\newblock In G.~Brassard, editor, {\em Advances in Cryptology -- {CRYPTO}'89},
  volume 435 of {\em Lecture Notes in Computer Science}, pages 239--252.
  Springer, New York, Aug. 1990.

\bibitem{ESORICS:SALY17}
S.-F. Sun, M.~H. Au, J.~K. Liu, and T.~H. Yuen.
\newblock {RingCT} 2.0: {A} compact accumulator-based (linkable ring signature)
  protocol for blockchain cryptocurrency monero.
\newblock In S.~N. Foley, D.~Gollmann, and E.~Snekkenes, editors, {\em
  ESORICS~2017: 22nd European Symposium on Research in Computer Security,
  Part~II}, volume 10493 of {\em Lecture Notes in Computer Science}, pages
  456--474. Springer, Cham, Sept. 2017.

\bibitem{sun2021survey}
X.~Sun, F.~R. Yu, P.~Zhang, Z.~Sun, W.~Xie, and X.~Peng.
\newblock A survey on zero-knowledge proof in blockchain.
\newblock {\em IEEE network}, 35(4):198--205, 2021.

\bibitem{wood2014ethereum}
G.~Wood.
\newblock Ethereum: A secure decentralised generalised transaction ledger,
  2019.

\bibitem{Wust2021}
K.~W{\"{u}}st, K.~Kostiainen, N.~Delius, and S.~Capkun.
\newblock {Platypus: A Central Bank Digital Currency with Unlinkable
  Transactions and Privacy Preserving Regulation}.
\newblock {\em Cryptology ePrint Archive}, 1:2947--2960, nov 2021.

\bibitem{Xu2020}
L.~Xu, L.~Chen, Z.~Gao, K.~Kasichainula, M.~Fernandez, B.~Carbunar, and W.~Shi.
\newblock {PrivateEx: Privacy preserving exchange of crypto-assets on
  blockchain}.
\newblock {\em Proceedings of the ACM Symposium on Applied Computing},
  8(20):316--323, mar 2020.

\bibitem{xu2023regulation}
L.~Xu, Y.~Zhang, and L.~Zhu.
\newblock Regulation-friendly privacy-preserving blockchain based on zk-snark.
\newblock In {\em International Conference on Advanced Information Systems
  Engineering}, pages 167--177. Springer, 2023.

\bibitem{FC:YSLAEZG20}
T.~H. Yuen, S.~Sun, J.~K. Liu, M.~H. Au, M.~F. Esgin, Q.~Zhang, and D.~Gu.
\newblock {RingCT} 3.0 for blockchain confidential transaction: Shorter size
  and stronger security.
\newblock In J.~Bonneau and N.~Heninger, editors, {\em FC 2020: 24th
  International Conference on Financial Cryptography and Data Security}, volume
  12059 of {\em Lecture Notes in Computer Science}, pages 464--483. Springer,
  Cham, Feb. 2020.

\bibitem{zhi2022ledgit}
X.~Zhi, Y.~Satsangi, S.~Moran, and S.~Eloul.
\newblock Ledgit: A service to diagnose illicit addresses on blockchain using
  multi-modal unsupervised learning.
\newblock In {\em Proceedings of the 31st ACM International Conference on
  Information \& Knowledge Management}, CIKM '22, page 5069–5073, New York,
  NY, USA, 2022. Association for Computing Machinery.

\end{thebibliography}

\newpage
\appendix

\section*{Appendix}

\section{Notations}
\begin{tabular}{llll}
$\group $& Group &  $\pactproofs^{\asset}$ & Proof of asset  \\ \\
$g$, $h$ & Generators  & $\pactproofs^{\consistency}$ & Proof of consistency \\ \\
%$\mathbb{G}$ & The group of points on the elliptic curve Secp256k1 over the field $F_P$ and order \\ & $N$ as recommended in \cite{brown2010sec} & & \\ \\ 
%$g,h$ & Point generators of $\mathbb{G}$, $h$ is the base point of $\mathbb{G}$ and $g$ is an additional random point generator  & & \\ \\ 
$\grouporder$ & Scalar field & $\pactproofs^{\equivalence}$ & Proof of equivalence\\ \\
%broadcast-tx &  Broadcaster's proposed (pending) transaction & & \\ \\ 
$p$ & Participant index & $\dataend$ & Endorsement data  \\ \\
$t$ & Transaction index &   $\pactpks$  &  Public key list \\ \\
$a$ & Asset index & $\pactvalue$ & Transaction value list  \\ \\
sender & The participant who executes $\pactspend$ & &  \\ \\
$T$, $\pacttxs$ & All transactions & \\ \\
$A$, $\mathcal{AS}$ & All assets & \\ \\
$P$& All participants & \\ \\
$\cell$ &  Entry in transaction row & & \\ \\
%indexed by asset $a$, participant $p$, and transaction $t$
$\pacttx$ & Cells representing a transaction & \\ \\
$txs$ & Set of $tx$ & \\ \\
$\pactcoinvalue$ & Transaction value & \\ \\
$\coinrand$ & Blinding factor / Randomness & \\ \\
$\pactpolicy$ & A predicate policy function & \\ \\
$\commitment$ &  Pedersen Commitment & \\ \\
$\commitment'$ & Complementary Pedersen Commitment & \\ \\
$\token$ &  Token& \\ \\
$\token'$ & Complementary Token & \\ \\
%  for signing a complementary blinding factor with public key
$\pactproofs$ & Set of proofs for a $\pacttx$ & & \\ \\ 
\end{tabular}

\section{Use-Cases Transaction Examples}
\label{appendix:usecases}

In this section we provide an example of construction of transaction commitments by the broadcaster and the setup of the ledger. We show examples of transaction for the experimental use cases discussed in Sec. \ref{ch:usecases}.
\subsection{Bond exchange and Coupon rate}
We assume a simple use case of 2 investors that invest in bonds, with 2-year maturity, 10\% yearly coupon rate, and a par value of 10\$ for each bond unit. In this exchange, there are several actors: custodian is the financial institution that safe keeps the asset of the investors (here USD) and can issue a digital coin/token that is pegged to fiat currency, USD. The bond issuer issues the bonds, and borrows USD. The broker is going to manage the deal, and the exchange between the investors and bond issuer.

Privacy assumptions:
\begin{itemize}
    \item The broker knows the details of the bond deal: How many bonds unit each investor buys.
    \item The custodian has no other knowledge except the amount of money it issues to the investors.
    \item The bond issuer does not know the individual investors, just the sum of USD asset received from many (two) investors.
    \item The investors do not know about other investors values in the $t$.
    \item All participants ($p$) in the ledger know the time that an event of transaction ($t$) happened on ledger, but the content is encrypted.
\end{itemize}
 
In PADL we create a ledger that includes all participants. There is an initialisation where issuers are capable to mint assets to the ledger ($tx0$). Let us assume the bond issuer data: bond $X$, 2-years maturity, and 100 units with value $1,000\$$, with yearly coupon rate of 10\%. Asset name: $X10\$2y10\%$.  The Custodian mint 3000\$, and the Bond issuer with mint 1000\$ and $300X$.  The following transactions that form the ledger do not include all the information that the user will see in the PADL ledger. We have omitted items, such as proofs and complimentary data, and we have kept only the commits ($cm$) and the tokens ($tk$) for clarity.

\begin{table}[h!]
\captionsetup{labelformat=empty}
\begin{center}
\caption{\textbf{2024 Jan}: $tx0$, initialisation}
\begin{tabular}{cccccc}
\hline
 Asset & Custodian & Bond Issuer & Broker & Investor M & Investor N \\ [0.5ex] 
 \hline
 0 & $g^{3000\$}h^{r0a0}$, $pk^{r0a0}$ & $g^{1000\$}h^{r0b0}, pk^{r0b0}$ & $g^{0\$}h^{r0c0}, pk^{r0c0}$ & $g^{0\$}h^{r0d0}, pk^{r0d0}$ & $g^{0\$}h^{r0e0}, pk^{r0e0}$\\
 1 & $g^{0X}h^{r0a1}$, $pk^{r0a1}$ & $g^{300X}h^{r0b1}, pk^{r0b1}$ & $g^{0X}h^{r0c1}, pk^{r0c1}$ & $g^{0X}h^{r0d1}, pk^{r0d1}$ & $g^{0X}h^{r0e1}, pk^{r0e1}$\\
\end{tabular}
\end{center}
\end{table}

The notations above, for $g, h, pk, r$ are explained in Sec. \ref{methods}. The $r$ (blending factors) is new for each $t$ (transaction) and participant, as indicated by their indices.

$tx1$ and $tx2$ are USD-tokens with one custodian for both investors. The custodian moves money to investors, to mint tokens of stable USD as requested by the investors ($2,000\$$ to each investor).

\begin{table}[h!]
\captionsetup{labelformat=empty}
\begin{center}
\caption{\textbf{2024 Jan}: $tx1$}
\begin{tabular}{cccccc}
\hline
 Asset & Custodian & Bond Issuer & Broker & Investor M & Investor N\\ [0.5ex] 
 \hline
 0 & $g^{-2,000\$}h^{r1a}$, $pk^{r1a}$ & $g^{0\$}h^{r1b}, pk^{r1b}$ & $g^{0\$}h^{r1c}, pk^{r1c}$ & $g^{2,000\$}h^{r1d}, pk^{r1d}$ & $g^{0\$}h^{r1e}, pk^{r1e}$ \\
\end{tabular}
\end{center}
\end{table}

\begin{table}[h!]
\captionsetup{labelformat=empty}
\begin{center}
\caption{\textbf{2024 Jan}: $tx2$}
\begin{tabular}{cccccc}
\hline
 Asset & Custodian & Bond Issuer & Broker & Investor M & Investor N\\ [0.5ex] 
 \hline 
 0 & $g^{-2,000\$}h^{r2a}$, $pk^{r2a}$ & $g^{0\$}h^{r2b}, pk^{r2b}$ & $g^{0\$}h^{r2c}, pk^{r2c}$ & $g^{0\$}h^{r2d}, pk^{r2d}$ & $g^{2,000\$}h^{r2e}, pk^{r2e}$ \\
\end{tabular}
\end{center}
\end{table}

Broker broadcast a $tx3$ on PADL, a single exchange of bonds unit with lending to the bond issuer.
\begin{itemize}
    \item Investor $M$ sends $1,000\$$ to the bond issuer and gets $100X$ units, 
    \item Investor $N$ gets $200X$ units for $2,000\$$, the bond issuer sees $3,000\$$ received in the $t$.
\end{itemize}

No one besides the broker can tell what the blending values ($r$’s) are or what are the invested \$/bond values of the other participants. They can only verify their values. For example, the issuer can only infer that the money it received in the $t$ is $3,000\$$, and the bond unit it sent is at a value of $300X$. 

\begin{table}[h!]
\captionsetup{labelformat=empty}
\begin{center}
\caption{\textbf{2024 Jan}: $tx3$, exchange}
\begin{tabular}{cccccc}
\hline
 Asset & Custodian & Bond Issuer & Broker & Investor M & Investor N\\ [0.5ex] 
 \hline
 0 & $g^{0\$}h^{r3a}$, $pk^{r3a}$ & $g^{3,000\$}h^{r3b}, pk^{r3b}$ & $g^{0\$}h^{r3c}, pk^{r3c}$ & $g^{-1,000\$}h^{r3d}, pk^{r3d}$ & $g^{-2,000\$}h^{r3e}, pk^{r3e}$ \\
  
 1 & $g^{0X}h^{r4a}$, $pk^{r4a}$ & $g^{-300X}h^{r4b}, pk^{r4b}$ & $g^{0X}h^{r4c}, pk^{r4c}$ & $g^{100X}h^{r4d}, pk^{r4d}$ & $g^{200X}h^{r4e}, pk^{r4e}$ \\

\end{tabular}
\end{center}
\end{table}

Broker also sends to the issuer of the bond, the future coupons $t$ so the issuer can broadcast these $t$ every year. Only the broker knows the values put for investor $B$ and $M$, the bond issuer knows the total, and can verify the rate but does not know about the individual values of the investors.

\begin{table}[h!]
\captionsetup{labelformat=empty}
\begin{center}
\caption{\textbf{2025 Jan}: $tx4$, Year 1 coupon }
\begin{tabular}{cccccc}
\hline
 Asset & Custodian & Bond Issuer & Broker & Investor M & Investor N \\ [0.5ex] 
 \hline
 0 & $g^{0\$}h^{r5a}$, $pk^{r5a}$ & $g^{-300\$}h^{r5b}, pk^{r5b}$ & $g^{0\$}h^{r5c}, pk^{r5c}$ & $g^{100\$}h^{r5d}, pk^{r5d}$ & $g^{200\$}h^{r5e}, pk^{r5e}$ \\
  
 1 & $g^{0X}h^{r6a}$, $pk^{r6a}$ & $g^{0X}h^{r6b}, pk^{r6b}$ & $g^{0X}h^{r6c}, pk^{r6c}$ & $g^{0}h^{r6d}, pk^{r6d}$ & $g^{0X}h^{r6e}, pk^{r6e}$ \\

\end{tabular}
\end{center}
\end{table}

\begin{table}[h!]
\captionsetup{labelformat=empty}
\begin{center}
\caption{\textbf{2026 - Jan}: $tx5$, Year 2 coupon}
\begin{tabular}{cccccc}
\hline
 Asset & Custodian & Bond Issuer & Broker & Investor M & Investor N \\ [0.5ex] 
 \hline
 0 & $g^{0\$}h^{r7a}$, $pk^{r7a}$ & $g^{-300\$}h^{r7b}, pk^{r7b}$ & $g^{0\$}h^{r7c}, pk^{r7c}$ & $g^{100\$}h^{r7d}, pk^{r7d}$ & $g^{200\$}h^{r7e}, pk^{r7e}$ \\
  
 1 & $g^{0X}h^{r8a}$, $pk^{r8a}$ & $g^{0X}h^{r8b}, pk^{r8b}$ & $g^{0X}h^{r8c}, pk^{r8c}$ & $g^{0}h^{r8d}, pk^{r8d}$ & $g^{0X}h^{r8e}, pk^{r8e}$ \\

\end{tabular}
\end{center}
\end{table}

The bond issuer keeps the encrypted $t$ and will send the first coupon $t$ to the ledger at the end of the year and the second at the end of the second year. The issuer cannot tell how the money is distributed between the investors, only that its total rate, 300\$ is the right amount recorded in these $t$.

\textbf{Proof of the coupon rate}. This is an additional proof to make sure broker and participants are honest in the coupon $t$, even though everyone may be fine with its money received. All participants in addition to the general proofs in PADL $t$, creates a proof of ratio equivalence. Each participant add this prove to the cell when they add the rest of the proofs of the cell.

Finally, in maturity of the bond issuer returns the money and investor return the bonds units. The Broker also gets 0.1\% fees of the total bonds for its hard work. The fees proof can be also added to the cell using the rate proof to ensure the Broker is paid exactly 0.1\% of everyone.

\begin{table}[h!]
\captionsetup{labelformat=empty}
\begin{center}
\caption{\textbf{2026 - maturity}: Transaction Row, $tx6$, return assets and fees}
\begin{tabular}{cccccc}
\hline
 Asset & Custodian & Bond Issuer & Broker & Investor M & Investor N \\ [0.5ex] 
 \hline
 0 & $g^{0\$}h^{r9a}$, $pk^{r9a}$ & $g^{-3,003\$}h^{r9b}, pk^{r9b}$ & $g^{6\$}h^{r9c}, pk^{r9c}$ & $g^{999\$}h^{r9d}, pk^{r9d}$ & $g^{1998\$}h^{r9e}, pk^{r9e}$ \\
  
 1 & $g^{0X}h^{r10a}$, $pk^{r10a}$ & $g^{300X}h^{r10b}, pk^{r10b}$ & $g^{0X}h^{r10c}, pk^{r10c}$ & $g^{-100X}h^{r10d}, pk^{r10d}$ & $g^{-200X}h^{r10e}, pk^{r10e}$ \\

\end{tabular}
\end{center}
\end{table}

\subsection{Settlement Bank Example}

The following transactions are exemplary of Sec. \ref{settlement_ledger_trusted}. Assume we have a settlement bank and two counter-party banks. In total the ledger has 3 participants. The two counter-party banks, Bank A and Bank B, wish to make a payment on assets. For example, Bank A wants to trade money market funds (short term debt) for USD-tokens from Bank B, to cover short term liquidity requirements. They can do this through a trusted participant (settlement bank) that will ensure that both parties have adequate funds to complete the transaction. The settlement party has to be able to verify the balance without having access to the amounts or assets traded. 

Privacy assumptions:
\begin{itemize}
    \item A trusted party ($p$) to act as the settlement bank and be able to audit the transactions ($t$).
    \item The two counterparty banks (Bank A and Bank B) will not share their private keys ($sk$) with the settlement party. 
\end{itemize}

Similarly to the \textbf{Bond Coupon Simple Example}, PADL creates a ledger with an initialisation line ($tx0$).  

\begin{table}[h!]
\captionsetup{labelformat=empty}
\begin{center}
\caption{\textbf{2024 Jan}: $tx0$}
\begin{tabular}{cccc}
\hline
 Asset & Settlement Bank (Issuer) & Bank A & Bank B  \\ [0.5ex] 
 \hline
 0 & $g^{0\$}h^{r0a}$, $pk^{r0a}$ & $g^{0\$}h^{r0b}, pk^{r0b}, pk^{r0b}_{I}$ & $g^{2000\$}h^{r0c}, pk^{r0c}, pk^{r0c}_{I}$ \\
 1 & $g^{0\$}h^{r1a}$, $pk^{r1a}$ & $g^{10MM\$}h^{r1b}, pk^{r1b}, pk^{r1b}_{I}$ & $g^{0\$}h^{r1c}, pk^{r1c}, pk^{r1c}_{I}$ \\
\end{tabular}
\end{center}
\end{table}

Compared to the previous example, we add an additional token ($tk_I)$ that signs $r_{t, a, p}$, by the public key to the tokens of the issuer $I$.  The initial $tx0$ is USD-tokens with Bank B and Bank A starts with market fund $10MM$. 
In the subsequent transaction ($tx1$), Bank A will send $10$ money market funds ($10MM$), with value $200\$$ each, for exchange for $2,000\$$ USD tokens. The settlement bank keeps the money market funds ($10MM$) until Bank A repays back the $2,000\$$ in \textbf{a months' time}.

\begin{table}[h!]
\captionsetup{labelformat=empty}
\begin{center}
\caption{\textbf{2024 Jan}: $tx2$}
\begin{tabular}{cccc}
\hline
 Asset & Settlement Bank & Bank A & Bank B \\ [0.5ex] 
 \hline
 0 & $g^{0\$}h^{r2a}$, $pk^{r2a}$ & $g^{2,000\$}h^{r2b}, pk^{r2b}, pk^{r2b}_{I}$ & $g^{-2,000\$}h^{r2c}, pk^{r2c}, pk^{r2c}_{I}$  \\
  1 & $g^{10MM}h^{r3a}$, $pk^{r3a}$ & $g^{-10MM}h^{r3b}, pk^{r3b}, pk^{r3b}_{I}$ & $g^{0}h^{r3c}, pk^{r3c}, pk^{r3c}_{I}$  \\
\end{tabular}
\end{center}
\end{table}
At the end of the month, Bank A sends to Bank B the $2,000\$$ value and the settlement bank releases the money market funds back to Bank A.

\begin{table}[h!]
\captionsetup{labelformat=empty}
\begin{center}
\caption{\textbf{2024 Feb}: $tx3$}
\begin{tabular}{cccc}
\hline
 Asset & Settlement Bank & Bank A & Bank B (Issuer) \\ [0.5ex] 
 \hline
 0 & $g^{0\$}h^{r4a}$, $pk^{r4a}$ & $g^{-2,000\$}h^{r4b}, pk^{r4b}, pk^{r4b}_{I}$ & $g^{2,000\$}h^{r4c}, pk^{r4c}, pk^{r4c}_{I}$  \\
  1 & $g^{-10MM}h^{r5a}$, $pk^{r5a}$ & $g^{10MM}h^{r5b}, pk^{r5b}, pk^{r5b}_{I}$ & $g^{0MM}h^{r5c}, pk^{r5c}, pk^{r5c}_{I}$  \\
\end{tabular}
\end{center}
\end{table}
In all cells, the settlement bank can extract the value using the commit and its corresponding $tk_I$.

\section{Security and Privacy Analysis}
\label{appendix:security}
The security and privacy guarantees of PADL is stated formally here. First, we recall the standard cryptographic definitions that are omitted in the main body for both the commitment scheme and NIZK. Then, proof is given for both Integrity and Anonymity properties of PADL.

\subsection{Additional Background}
\begin{definition}[Properties of Commitment Scheme]
A non-interactive scheme $(\commitgen, \commit)$ is called a commitment scheme if it satisfies the following properties.

\textbf{Hiding. }Let the advantage $\hidingadv$ of breaking the binding property be the probability defined below. A scheme satisfies hiding if for all PPT adversary $\adv$, there exists a negligible function $\negl$ such that
\begin{align*}
\abs{ \Pr\left[
\begin{array}{cc}
    \multirow{2}{*}{$\adv(\commitment)=b:$}  
    & \commitkey \gets \commitgen(\secparam); (m_0,m_1)\gets \adv(\commitkey)\\
    & b\sample \bin; \commitment \gets \commit(\commitkey, m_b)
\end{array}
\right] - \frac{1}{2}} \leq \negl
\end{align*}

\textbf{Binding. } Let the advantage $\bindingadv$ of breaking the binding property be the probability defined below. A scheme satisfies binding if for all PPT adversary $\adv$, there exists a negligible function $\negl$ such that
\begin{align*}
\Pr\left[
\begin{array}{cc}
    \multirow{2}{*}{$m_0 \neq m_1 \wedge \commit(m_0,r_0) = \commit(m_1,r_1) :$}  
    & \commitkey \gets (\commitgen(\secparam); \\
    & (m_0,r_0,m_1,r_1)\gets \adv(\commitkey)
\end{array}
\right] \leq \negl
\end{align*}
\label{definition:com_properties}
\end{definition}

\begin{definition}[Properties of NIZK]
\label{definition:nizk_property}
A NIZK proof system has the following properties.

\noindent\textbf{Completeness.} For all security parameters $\secpar$, all statements $\stmt \in \lang_{\relation}$, all witness $\wit$ with $\relation(\stmt, \wit)=1$, let $\crs \gets \nizksetup(\secpar)$ and $\nizkproof \gets \nizkprove(\crs,\stmt,\wit)$, it holds that $\nizkverify(\crs,\stmt,\nizkproof)=1$.

\noindent\textbf{Computational Soundness.} The proof system is computationally sound if for all PPT adversaries $\adv$, there exists a negligible function $\negl$ such that the following holds:
%The advantage $\sf{Sound}^{\adv}$ of an adversary $\adv$ in breaking the  soundness of the proof system is the probability defined below. 
\begin{align*}  
\Pr\left[
\begin{array}{c}
\nizkverify(\crs, \stmt, \nizkproof)=1: \\
\crs \gets \nizksetup(\lambda); \\
(\nizkproof, \stmt) \gets \adv(\crs);  \\
\stmt \notin \lang_{\relation}
\end{array}
\right] \leq \negl
\end{align*} 
%^{\nizkverify(\crs,\cdot,\cdot)}

\noindent\textbf{Zero-Knowledge} The advantage $\sf{ZK}^{\adv}(\secpar)$ of an adversary $\adv$ in  breaking the zero-knowledge property is the probability defined below. The proof system is zero-knowledge if there exists simulator algorithms ($\simulator_1, \simulator_2$) such that for all PPT adversaries $\adv$, there exists a negligible function $\negl$ such that $\sf{ZK}^{\adv}(\secpar) \leq \negl$.
\begin{align*}
\Big|&\Pr\left[\begin{array}{cc}
\multirow{4}{*}{$\adv(\nizkproof^{\ast}) = 1~:$}  & \crs \gets \nizksetup(\secpar); \\
  & (\stmt, \wit) \gets \adv(\crs); \\
  &   \relation(\stmt, \wit) = 1;   \\
   &  \nizkproof^\ast \gets \nizkprove(\crs, \\
   &  \stmt,\wit)
\end{array}
\right]   - \\ 
&\Pr\left[\begin{array}{cc}
\multirow{4}{*}{$\adv(\nizkproof^{\ast}) = 1~:$} &  (\crs, {\sf st}) \gets \simulator_1(\secpar);  \\
 &  (\stmt, \wit) \gets \adv(\crs); \\
  &   \relation(\stmt, \wit) = 1; \\
  &    \nizkproof^\ast \gets \simulator_2(\crs,\stmt,{\sf st}) 
\end{array}\right]\Big| 
\end{align*}
\end{definition}

% Theorem from Meuler
% Will be moved to appendix later.
\subsubsection{Generalized Zero-Knowledge Proof for Group Homomorphism.}
We rely on the following generalized special honest-verifier zero-knowledge proof of knowledge ($\zkpok$) $\Sigma$-protocol for preimage of a group homomorphism by Maurer \cite{AFRICACRYPT:Maurer09}.
\begin{theorem}
\label{theo:maurer}
Let $(G_1, \opone)$ and $(G_2, \optwo)$ be two groups, and $\phi: G_1 \rightarrow G_2$ be a group homomorphism such that $\phi(a \opone b) = \phi(a) \optwo \phi(b)$. Let $\ell \in \mathbb{Z}$, $u \in G_1$, and $z = \phi(x)$ are known such that:

(1) $\forall c_1,c_2 \in \challengespace$ with $c_1 \neq c_2$, gcd($c_1 - c_2,\ell$) = 1

(2) $\phi(u) = z^\ell$

then a $\zkpok$ for a given $z$, of a value $x$ such that $z = \phi(x)$ can be constructed by having the prover samples $ k \sample  G_1$ and computes $t := \phi(k)$, verifier then generates a challenge $ c \sample \challengespace$, prover then outputs $ r := k \opone x^c$, and finally the verifier checks 
$\phi(r) {=^?} t \optwo z^c$.
% $\phi(r) \overset{?}{=} t \optwo z^c$.
\end{theorem}

Using the theorem defined above, we can now show that proof of consistency defined in Section \ref{sec:padlcon} is a $\zkpok$ protocol. 

\begin{theorem}
Let $\group$ to be $p$-prime ordered group and $g,h,\pactpk \in \group$. Then, proof of consistency defined in Section \ref{sec:padlcon} is a ZKP protocol for the language $\lang_{\relation_{\consistency}} := \{(\commitment,\token) : \exists \pactcoinvalue,\coinrand\textnormal{ s.t. }\commitment = g^{\pactcoinvalue}h^{\coinrand} \wedge \token = \pactpk^{\coinrand} \}$.
\end{theorem}
\begin{proof}
Let $G_1 = \mathbb{Z}_p \times \mathbb{Z}_p$ with its group operation being component-wise addition and $G_2 = \group \times \group$ with its group operation being component-wise. Then $\phi(\pactcoinvalue,\coinrand) := (g^{\pactcoinvalue} h^{\coinrand}, \pactpk^{\coinrand})$ is a group homomorphism as $\phi(\pactcoinvalue_1,\coinrand_1) \optwo \phi(\pactcoinvalue_2,\coinrand_2) = (g^{\pactcoinvalue_1+\pactcoinvalue_2}h^{\coinrand_1+\coinrand_2}, \pactpk^{\coinrand_1+\coinrand_2}) = \phi(\pactcoinvalue_1 + \pactcoinvalue_2,\coinrand_1+\coinrand_2)$. 
Let $\ell=p$ and $u=(0,0)$, $\forall z \in G_2$, $z^\ell = (1,1) = \phi(u)$ satisfied the requirements for Theorem \ref{theo:maurer}. We can now apply Theorem \ref{theo:maurer} to obtain a ZKP protocol for $\lang_{\relation_{\consistency}}$ and then apply Fiat-Shamir transform to obtain the corresponding NIZK.
\qed
\end{proof}

\subsection{Model and Analysis}
\label{appendix:adv}

The state of PADL ledger is interacted with $\exp$ oracles described in Definition \ref{definition:exp_oracle}. As discussed in Section \ref{sec:eval}, integrity property is derived to maintain the invariant defined in Definition \ref{def:invariant}. Integrity property ensures that for each transaction, it always preserves asset balance, transaction balance and ownership endorsement requirements.

\begin{definition}[Oracles] The adversary will be given oracles 
%$\experiment := \{\expaddacc,\expspend,\expcorrupt,\exppolicy,\expledger\}$
with state of transactions $\pacttxlist$, account public key list $\pactpklist$, account key pair list $\expaccountlist$, invalid transaction list $\exptxlist$, and corrupted key pair list $\expcorruptlist$ defined as below depending on the experiments:
	\begin{itemize}
	\item{$\expaddacc\leftbrc\rightbrc$:} the oracle executes $(\pactsk,\pactpk) \gets \pactkeygen(\secparam)$ 
 , registers a new account under $\pactpk$ by updating its state with $\pactpklist = \pactpklist \bigcup \{\pactpk\}$, appends $(\pactpk,\pactsk)$ to $\expaccountlist$, and returns address $\pactpk$.
   \item{$\exppost\leftbrc \pacttx, \pactproofs \rightbrc$:} on input a transaction $\pacttx$ and a validity proof $\pactproofs$, if $\pactverify(\allowbreak\pacttx,\allowbreak\pactproofs,\pacttxlist,\allowbreak \pactpklist)=0$ or $(\pacttx,\pactproofs) \in \exptxlist$, the oracle returns failure bit $b=0$ without appending the ledger, else the oracle appends its ledger $\pacttxlist$ with $\pacttx$ and returns success bit $b = 1$. 
    \item{$\expspend\leftbrc\pactvalue\rightbrc$:} on input a transaction amount list $\pactvalue$, 
    %for all public key $\pactpk_p \in \pactspk$ where spending account list $\pactspk$ is parsed from $\pactvalue$,   the oracle retrieves $ (\pactpk_p, \pactsk_p) \in \expaccountlist$,  $\forall p$ the oracles temporarily set policy $\pactpolicy_p$ to True for the current execution,
    executes $(\pacttx, \pactproofs) \gets \pactspend\leftbrc \pactvalue\allowbreak, \allowbreak\pacttxlist\allowbreak ,\allowbreak \pactpklist\rightbrc$, 
    %restore $\pactpolicy_p$, executes 
    $(b) \gets \exppost(\pacttx, \pactproofs)$, and the oracle returns $\ground$ if $b = 0$, else returns $(\pacttx, \pactproofs)$.
    % appends $\pacttx$ to $\expspendlist$ before
    \item{$\expcorrupt\leftbrc \pactpk \rightbrc$:} on input an account public key $\pactpk$, the oracle retrieves account key pair $(\pactpk,\pactsk) \in \expaccountlist$, appends $(\pactpk)$ to $\expcorruptlist$ and returns $\pactsk$. %$(\pactpk,\pactsk)$ 
    \item{$\exppolicy(p,\mathcal{F})$:} on input a party identifier $p$, and a predicate function $\mathcal{F}$, the oracle executes $\pactpolicy_p := \mathcal{F}$.
    \item{$\expledger()$}: the oracle returns $(\pactpklist,\pacttxlist)$.
\end{itemize}
\label{definition:exp_oracle}
\end{definition}

\begin{definition}[Integrity] It holds that every PPT adversary $\adv$ has at most negligible advantage in the following experiment, where we define the advantage as $\AdvBalance := \prob{\ExpBalanceFull(\secpar) = 1}$ .
    \procb{$\ExpBalanceFull$}{
        pp \gets \pactsetup(\secparam) 
        \\ \oracle := \{\expaddacc, \exppost, \expspend, \expcorrupt, \exppolicy, \expledger \}\\
        (\pacttx, \pactproofs) \gets \adv^{\oracle}(\publicpar) \\
        %\textnormal{spending accounts } \pactspk \textnormal{ is retrieved from } \pacttx \\
        (\pactpks, \pacttxs) \gets \expledger() \\
        \textnormal{return 1 if }\pactverify(\pacttx, \pactproofs, \pacttxs, \pactpks) = 1\\
        % , \pacttx \notin \expspendlist
        \textnormal{and one of the following hold for any asset $a$ in the transaction $\pacttx$:} \\
        (i)  \exists p \textnormal{ s.t. } \sum_{i=1}^{|\pacttxs|}{\pactcoinvalue_{i,p,a}} < \abs{\pactcoinvalue_{\pacttx,p,a}} \textnormal{ when $\pactcoinvalue_{\pacttx,p,a}$ is negative },\\
        (ii) \exists p \textnormal{ s.t. } \pactpk_{p} \notin \expcorruptlist, \pactproofs^{\asset}_{\pacttx,p,a}\ and\ \pactproofs^{\equivalence}_{\pacttx,p,a} \textnormal{ are not generated by the challenger,} \\
        (iii) \sum_{i=1}^{|\pactpks|}{\pactcoinvalue_{\pacttx,i,a}} \neq 0.
    }
    \label{definition:exp_balance}
\end{definition}

\begin{theorem}
\label{theorem:balance}
% Let us define the advantage of an adversary $\adv$ in winning the experiment in Definition~\ref{definition:exp_balance} as:
% $$\AdvBalance^{\adv} := \prob{\ExpBalanceFull_{\pact}(\secpar) = 1}$$

% For the PADL defined in Table \ref{table:padl_algo}, we have:
% \begin{align*}
% \AdvBalance^{\adv} \leq 
% \end{align*}

% Informally, PADL achieve integrity property as long as the proofs used in the transaction scheme is .
If the underlying proof system is sound and zero-knowledge, and the commitment scheme used is binding, then PADL is integrity preserving.
\end{theorem}
\begin{proof}
Suppose there exists an efficient adversary $\adv$ that can break the integrity property of PADL with non-negligible probability, we exhaust three cases where (i) proof of asset asserts enough account balance even though there is not enough account balance to execute the transaction 
%which is prevented due to sound range proof 
(ii) proof of asset is generated without account ownership (secret key) 
%which is prevented due to proof of knowledge, or 
(iii) proof of balance for commitment zero out to group identity despite the committed value does not zero out.
%which is prevented due to binding property of commitment scheme.

\textbf{Case 1 Asset Overspending: } Given a transaction history $\pacttxs$, current transaction $\pacttx$, and a party $p$, asset balance states that
$\sum_{i=1}^{|\pacttxs|}{\pactcoinvalue_{i,p,a}} + \pactcoinvalue_{\pacttx,p,a} < 2^N $ where $2^N$ is some predetermined positive limit. $\adv$ wins in this case means $\sum_{i=1}^{|\pacttxs|}{\pactcoinvalue_{i,p,a}} + \pactcoinvalue_{\pacttx,p,a} \geq 2^N $ for some party p. Following the soundness of equivalence proof and the binding property of commitment scheme, the value of $ \sum_{i=1}^{|\pacttxs|}\pactcoinvalue_{i,p,a} + \pactcoinvalue_{\pacttx,p,a}$ is committed under $\commitment_{\prod}'$.
We now construct a reduction that breaks the soundness of the proof system used for the range proof in proof of asset. The reduction simulates the whole system honestly. It follows that $\pactproofs^{\asset}_{p}$ is a proof for a false statement. Thus, by returning $((\commitment_{\prod}', 2^N),\pactproofs^{\asset}_{p})$, the reduction will break the soundness property of the proof system for the range proof used in proof of asset which contradicts the assumed soundness of the proof system used.

\textbf{Case 2 Endorsement Forgery: } 
Let $\pactproofs^{\asset}$ and $\pactproofs^{\equivalence}$ be the proofs that are not generated by the challenger and $\pactpk$ be the corresponding public key account. 
Following the discrete log hardness of $h^{\pactsk}$, zero-knowledge property of $\pactproofs^{\asset}$ and $\pactproofs^{\equivalence}$ for previous transactions, if $\pactpk$ is never queried to the key retrieval oracle, then the challenger runs $\adv$ until the same inputs are queried twice to the random oracle with two different proof outputs given by $\adv$ where the proofs are forged for the same public key account satisfying the above criteria.
From here, run the extractor algorithm from the Schnorr discrete log ZKP to obtain a solution for discrete log contradicting the assumed hardness of discrete log.

\textbf{Case 3 Transaction (Im)Balance: } Given a set of commitments $\mathcal{CM} := \{\commitment_1,\commitment_2,...,\commitment_{|\pactpks|}\}$ with the corresponding set of values $\pactvalue := \{\pactcoinvalue_1,\pactcoinvalue_2,...,\pactcoinvalue_{|\pactpks|}\}$ and randomness used $\mathcal{R}:= \{\coinrand_1, \coinrand_2, ..., \coinrand_{|\pactpks|}\}$ such that $\commitment_{\prod} := \prod_{\commitment \in \mathcal{CM}} \commitment \allowdisplaybreaks = \allowdisplaybreaks g^{\pactcoinvalue_1+\pactcoinvalue_2+...+\pactcoinvalue_{|\pactpks|}} \allowdisplaybreaks \cdot  \allowdisplaybreaks h^{\coinrand_1+ \coinrand_2+ ...+ \coinrand_{|\pactpks|}}=g^{0}h^{0}=1$. Let $\pactvalue^{\ast}$ and $\mathcal{R}^{\ast}$ be different sets of values and randomness that one used in the transaction such that $\pactcoinvalue^{\ast}_1+\pactcoinvalue^{\ast}_2+...+\pactcoinvalue^{\ast}_{|\pactpks|} = \pactcoinvalue_{\sum}^{\ast} \neq 0$ and $\coinrand_{\sum}^{\ast} = \sum_{\coinrand^{\ast} \in \mathcal{R}^{\ast}} \coinrand^{\ast}$, then  $g^{\pactcoinvalue^{\ast}_1+\pactcoinvalue^{\ast}_2+...+\pactcoinvalue^{\ast}_{|\pactpks|}} \allowdisplaybreaks \cdot  \allowdisplaybreaks h^{\coinrand^{\ast}_1+ \coinrand^{\ast}_2+ ...+ \coinrand^{\ast}_{|\pactpks|}} = g^{\pactcoinvalue_{\sum}^{\ast}}g^{x ( \coinrand_{\sum}^{\ast} ) } = g^0 = 1 $  which means $x =  -(\pactcoinvalue_{\sum}^{\ast}) / \coinrand_{\sum}^{\ast}$ where  $g^x = h$ and $\coinrand_{\sum}^{\ast} \neq 0$. Note that knowledge of $\pactcoinvalue^{\ast}, \coinrand^{\ast}$ is used to generate $\pactproofs^{\consistency}$. 
%Thus, $dlog_g h$ is known contradicting the assumed hardness of discrete log assumption.
The challenger always rewinds $\adv$ once per $\pactproofs^{\consistency}$ generation to obtain two different proof outputs with the same inputs and runs the extractor algorithm to obtain $\pactcoinvalue^{\ast}, \coinrand^{\ast}$ which allows the computation of discrete log when the above mentioned criteria is met. Thus, contradicting the assumed hardness of discrete log assumption.
%Therefore, we have a set of commitment $\mathcal{CM}$ with two possible opening sets of opening value $\pactvalue,\pactvalue^{\ast}$ with corresponding $\mathcal{R}, \mathcal{R}^{\ast}$ contradicting assumed the hardness of discrete log from the binding property of the Pedersen commitment scheme. [TODO give a better reasoning flow, see commented line]
% Since knowledge of v, r is the witness to generate the ZKP proof, by the proof of knowledge of ZKP we know that v, r is known. Therefore, calculating discrete log is possible.

With the case exhaustion, the adversary has no way to win the game anymore. Therefore, we proved the theorem. \qed
\end{proof}

% \YS{for all assets instead of any?}
% \YS{We should say here intuitively what does this definition and experiment mean. Also true for most of the other definitions}
% \wz{yes i will add in some text that describe it. It is any asset not all asset and also it is one of the conditions is violated and not all. IE there exists. Not all because we want that the adversary win or our scheme consider broken if any of the invariant is broken for any of the assets.}

The (k)-anonymity of $\pact$ prevents an adversary from distinguishing transaction amounts and type of assets for accounts that the adversary does not own among a set of spenders, receivers, and neutral parties. The adversary also will not able to distinguish between spenders, receivers and neutral parties as the transaction amount and its sign are hidden. This property is captured by the Anonymity defined in Definition \ref{definition:exp_anon}. Informally, $\pact$ safeguards the anonymity of all participants (including spenders and recipients) among the selected set. In a deployed system, the adversary would be able to learn some information (such as insufficient balance) about the transaction if the adversary is freely allowed to choose the transaction amount details that an honest party will always accept. This situation is prevented in the model by requiring the adversary to distinguish transactions without the side effect of the transaction being queryable. In the post-challenge query phase, the adversary is not allowed to query the account secret key anymore to exclude the trivial case of distinguishing by extracting the transaction value directly. Trivial cases such as differences in adversary accounts and differences in account policy enforcement for the two given inputs are also excluded.

\begin{definition}[Anonymity]
It holds that every PPT adversary $\adv$ has at most negligible advantage in the following experiment, where we define the advantage as $\AdvAnon := | \prob{\ExpAnonFull(\secpar) = 1} - 1/2 |$. 
    \procb{$\ExpAnonFull$}{
        pp \gets \pactsetup(\secparam)
        \\ \oracle := \{\expaddacc, \exppost, \expspend, \expcorrupt, \exppolicy, \expledger \}\\
        (\pactvalue_0, \pactvalue_1) \gets \adv^{\oracle}(\publicpar) \\
        b \sample \bin \\
        (\pactpks, \pacttxs) \gets \expledger() \\
        \ret \ground\ \textnormal{if}\ \exists p  \textnormal{ s.t. $\pactpk_p \in \expcorruptlist, \pactcoinvalue_p \in \pactvalue_0, \pactcoinvalue'_p \in \pactvalue_1, \pactcoinvalue_{p} \neq \pactcoinvalue_{p}'$}\\
        \ret \ground\ \textnormal{if}\ \exists p \textnormal{ and } \exists {b^*} \in \{0,1\}\ \textnormal{s.t. } \pactcoinvalue_{p} \in \pactvalue_{b^*} \wedge \pactcoinvalue_{p}' \notin \pactvalue_{1-{b^*}}\\
        \ret \ground\ \textnormal{if}\ \exists p\ \textnormal{s.t. } \pactcoinvalue_{p} \in \{ \pactvalue_0 \cup \pactvalue_1\},  \pactpolicy_p(\pactcoinvalue_p,\pactaux_p) \neq 1 \\
        %\textnormal{$ (Here modify to both V pol is true)$ For all $p$ where $\pactpk_{p} \in \pactspk$, temporarily set $\pactpolicy_p :=$  True}  \\
        (\pacttx, \pactproofs) \gets \pactspend(\pactvalue_b, \pacttxs, \pactpks) \\
        \exptxlist := \exptxlist \cup (\pacttx, \pactproofs) \\
        % s_b \gets \ledgerpost(\pacttx, \pactproofs) \\
        b' \gets \adv^{\oracle \setminus \{\expcorrupt\}}(\publicpar, \pacttx,\pactproofs) \\
        \textnormal{return 1 if}\ b' = b\ 
        %and\ (ii)\ \forall k \in (R_0 \cup R_1 \cup {{\pactspk}_0} \cup {\pactspk}_1), k \notin \expcorruptlist \textnormal{ and } (k, \cdot) \in \expaccountlist
        % \textnormal{ and } (iii)\ s_b = 1
    }
	\label{definition:exp_anon}
\end{definition}

\begin{theorem}
\label{theorem:anon}
Let us define the advantage of an adversary $\adv$ in winning the experiment in Definition~\ref{definition:exp_anon} as:
$$\AdvAnon^{\adv} := | \prob{\ExpAnonFull_{\pact}(\secpar) = 1} - 1/2 |$$

For the PADL defined in Table \ref{table:padl_algo}, we have:
\begin{align*}
\AdvAnon^{\adv} \leq (\sum_{a=1}^{\abs{\assetset}}{|\pactvalue_a^{\ast}|}) \cdot \hidingadvped + 4(\sum_{a=1}^{\abs{\assetset}}{|\pactvalue_a^{\ast}|}) \cdot \sf{ZK}^{\adv}(\secpar)
\end{align*}

\end{theorem}

\noindent Informally, PADL has anonymity as long as the proofs used in the transaction scheme is zero-knowledge and the commitment scheme is hiding.

\begin{proof}
The proof proceeds via a series of game changes. The first game is exactly the original game defined and the last game completely hide the information about $b$. 
% and $|\expcorruptlist|$ be the length of corrupted participants.
%and $\pactconsent$ algorithms 

\noindent\textbf{$\pcgame_0:$} This is exactly the original game defined in Definition \ref{definition:exp_anon}. \\
{ } \\
\noindent\textbf{$\pcgame_1:$} In the $\pactspend$ algorithm, the challenger replaces all transaction amount values $\pactcoinvalue_i$ in $\pactvalue$ (except values belonging to accounts in corrupted list $\expcorruptlist$) with a random value $\pactcoinvalue_i^{\ast}$ while making sure entries in the modified transaction summed to zero and it does not spend over the account balance limit.
Note that transaction values $\pactcoinvalue$ for adversary corrupted accounts in both $\pactvalue_0$ and $\pactvalue_1$ is the same to prevent trivial cases of distinguishing. 
%The challenger also set the non-corrupted account policy $\pactpolicy_p$ to true for this transaction.
Let $\pactvalue^{\ast} \subseteq \pactvalue$ be the sublist containing only the modified values and let $\mathcal{AS}$ be the set of assets in the transaction. 
%To ensure proof of balance is true (similar to that of additive sharing of 0), the challenger set the last value of the transaction to be the absolute value of the sum of all previous values, we have $\pactcoinvalue_{|\pactvalue^{\ast}|}^{\ast} := \sum_{i=1}^{|\pactvalue^{\ast}|-1} \abs{\pactcoinvalue_i^{\ast}} $.
The distribution of the resulting commitments is the same from $\adv$'s point of view. 
The game is indistinguishable from $\pcgame_0$ by the commitment hiding property, it follows that
$$|\Pr[\pcgame_1(\adv)=1] - \Pr[\pcgame_0(\adv)=1]|\leq (\sum_{a=1}^{\abs{\assetset}}{|\pactvalue_a^{\ast}|}) \cdot \hidingadvped$$
\noindent\textbf{$\pcgame_2:$} Let ($\simulator_0, \simulator_1$) be the simulator for the zero-knowledge experiment of the underlying proof system, and $|\pactpks|$ be the number of participants. During the challenge phase where two lists of potentially different $\pactvalue_1, \pactvalue_2$ is given by the adversary, instead of computing the proof normally using $\nizkprove_{\lang_{\relation}}(\crs,\stmt,\wit)$, we use simulator $\simulator_2(\crs, \stmt, \st)$ where $(\crs, \st) \gets \simulator_1(\secpar)$. Since all commitments and tokens are generated honestly as well as spending amounts are ensured to be lower than account balance, $\stmt \in \lang$ for all languages used ($\lang_{\relation_{\consistency}},\lang_{\relation_{\consistency'}}, \lang_{\relation_{\equivalence}},\lang_{\relation_{\asset}}$). The game is indistinguishable from $\pcgame_1$ by the zero-knowledge property of the proof system, it follows that
$$|\Pr[\pcgame_2(\adv)=1] - \Pr[\pcgame_1(\adv)=1]|\leq 4(\sum_{a=1}^{\abs{\assetset}}{|\pactvalue_a^{\ast}|}) \cdot \sf{ZK}^{\adv}(\secpar) $$
%Note that we prevented side effect non-zero entries of the transactions to be observed from querying oracles during post-challenge phase. 
Note that now $(\pacttx,\pactproofs)$ is generated independent of
%of observable difference caused by non-zero entries from 
$\pactvalue_0, \pactvalue_1$ and secret keys of $\pactpks$. 
Therefore, the spending with $(\pacttx,\pactproofs)$ in this game leaks no information about $b$ and $\Pr[b' = b]= 1/2$.
\\
\\
\qed
\end{proof}

\section{Threat Model}
PADL assumes that participants can be malicious by tying to either steal, hide or modify values in assets or transaction. Participants can also collaborate to collude and hide transaction information from auditors, or try to break the consistency of a distributed ledger. We identify several points for the threat model and discuss their risk and resolutions.  

\subsection{Ledger consistency}
We obtain the state of the ledger with each transaction, as the hash of all commits and tokens of the ledger. In order to verify previous transactions, one can re-calculate the hash of the ledger to verify that transactions are not omitted. It is also reasonable to assume that one cannot generate a new ledger, by `picking' transactions, as participants' proof of assets would fail if a transaction is omitted. Removing last transactions is possible, but for that all participants need to agree on. Though if the ledger is deployed on smart contract, it would also require changing the history of the blockchain service which is down to the security level of the service. 

\subsection{Cryptography Setup}
To ensure that $log_{g}h$ is hard, all parties could contribute to the randomness of $h$. For example, each party can choose $h^{r_p}$ on the curve and the $\sum_p h^{r_p}$ is used as the ledger generator point, with no-trust.

\subsection{Quantum attack}
The current commit scheme is not considered quantum safe, but the combined system can be adapted to PQC primitives. This is intended to be explored in future study, for example by using lattice based cryptography for commitment scheme and $\Sigma$-protocol ZKPs. 

\subsection{Other Attacks Consideration}
A replay attack by copying a transaction is improbable because the state of the ledger used to verify the proofs would be different. Assets in a transaction cannot be separated into individual transactions, as the hash of the transaction intent would be different. Similarly, generated proofs cannot be separated from the transaction it is proved for as the transaction identifier is a part of the statement. In addition, the broadcaster cannot append a transaction by itself without endorsement from account owners, as the proof of assets and complementary commits and tokens would be missing.

\end{document}